\documentclass[authoryear, preprint, 11pt]{elsarticle}
\usepackage{amsmath, amsthm, amssymb}

\usepackage{graphicx}
\usepackage{fullpage}
\usepackage{xcolor}
\usepackage{natbib}
\usepackage{array,booktabs,makecell,multirow}
\usepackage[utf8]{inputenc}
\usepackage[T1]{fontenc}
\usepackage{float}
\usepackage{multirow}
\usepackage{indentfirst}
\usepackage{csquotes}
\usepackage[english]{babel}	
\usepackage{subfig}
\usepackage[skip = 13pt]{caption}
\RequirePackage[colorlinks,citecolor=blue,urlcolor=blue]{hyperref}
\usepackage{xcolor}
\usepackage{tikz}

\newtheorem{lemma}{Lemma}

\newtheorem{theorem}{Theorem}
\newtheorem{prop}{Proposition}

\DeclareMathOperator{\Tr}{Tr}
\DeclareMathOperator{\MSE}{MSE}

\DeclareMathOperator*{\argmin}{argmin}
\DeclareMathOperator{\GCV}{GCV}

\DeclareMathOperator*{\sign}{sign}
\definecolor{lm}{rgb}{1, 0, 1}
\definecolor{dgr}{rgb}{0.4980392, 0.4980392, 0.4980392}

\definecolor{c1}{rgb}{0,  1, 0}
\definecolor{c2}{rgb}{0.2470588,  0.7490196, 0.2470588}
\definecolor{c3}{rgb}{0.4980392, 0.4980392, 0.4980392}
\definecolor{c4}{rgb}{0.7490196, 0.2470588, 0.7490196}
\definecolor{c5}{rgb}{1, 0, 1}

\DeclareRobustCommand\full  {\tikz[baseline=-0.6ex]\draw[c3, thick] (0,0)--(0.57,0);} % Median
\DeclareRobustCommand\dotted{\tikz[baseline=-0.6ex]\draw[c2,thick,dashdotted, dash pattern=on 5pt off 2pt on \the\pgflinewidth off 2pt] (0,0)--(0.57,0);} % 0.3
\DeclareRobustCommand \dottedd {\tikz[baseline=-0.6ex]\draw[c2, thick, densely dotted] (0,0)--(0.57,0);} % 0.3
\DeclareRobustCommand\denselydashed{\tikz[baseline=-0.6ex]\draw[c5, thick, dash pattern={on 7pt off 1.5pt}] (0,0)--(0.57,0);} % 0.9
\DeclareRobustCommand\dotdash {\tikz[baseline=-0.6ex]\draw[c4,thick] (0,0)--(0.62,0);} %0.7
\DeclareRobustCommand\dotdashh {\tikz[baseline=-0.6ex]\draw[c4,dash dot] (0,0)--(0.62,0);} %0.7
\newtheorem{example}{Example}

\begin{document}

\begin{frontmatter}

\title{Robust penalized spline estimation with difference penalties}
\author[1]{Ioannis Kalogridis} 
\ead{ioannis.kalogridis@kuleuven.be}

\author[1]{Stefan Van Aelst\corref{cor1}} 
\ead{stefan.vanaelst@kuleuven.be}

\address[1]{Department of Mathematics, KU Leuven (University of Leuven), Celestijnenlaan 200B, 3001 Leuven, Belgium}

\cortext[cor1]{Corresponding author.}

\begin{abstract}
Penalized spline estimation with discrete difference penalties (P-splines) is a popular estimation method for semiparametric models, but the classical least-squares estimator is highly sensitive to deviations from its ideal model assumptions. To remedy this deficiency, a broad class of P-spline estimators based on general loss functions is introduced and studied. Robust estimators are obtained by well-chosen loss functions, such as the Huber or Tukey loss function. A preliminary scale estimator can also be included in the loss function. It is shown that this class of P-spline estimators enjoys the same optimal asymptotic properties as least-squares P-splines, thereby providing strong theoretical motivation for its use. The proposed estimators may be computed very efficiently through a simple adaptation of well-established iterative least squares algorithms and exhibit excellent performance even in finite samples, as evidenced by a numerical study and a real-data example.
\end{abstract}


\end{frontmatter}
\section{Introduction}
\label{sec:1}
Based on data $(x_1, Y_1), \ldots, (x_n, Y_n)$ with fixed $x_i$, which we assume to be in $[0,1]$ without loss of generality, the classical nonparametric regression model posits the relationship
\begin{align}
Y_i = f_0(x_i) + \epsilon_i, \quad (i=1, \ldots, n),
\end{align}
where $f_0:[0,1] \to \mathbb{R}$ is a smooth regression function to be estimated from the data. The errors $\epsilon_i$ are independent and identically distributed noise terms, which are often assumed to have zero mean and finite variance, but we will not need this assumption for well-chosen loss functions (see Section 3).  

Nonparametric regression has been a popular field of statistics for many years now and many methods of estimating $f_0$ have been proposed. Piecewise polynomial estimators (splines) still occupy a prominent place.  The first class of such estimators, referred to as \textit{smoothing}  splines, was introduced in the mid-70s, see \citep{Wahba:1990, Green:1994} for a review, and almost dominated the literature until the mid 80s when their position was challenged by lower rank, computationally cheaper alternatives, termed \textit{regression} and \textit{penalized} splines \citep{Wegman:1983, O:1986}. The former class of estimators uses a small number of strategically placed knots without any penalty, while the latter class usually employs a large number of knots, although still fewer than the number of data points, in combination with a quadratic roughness penalty. The kind of penalty employed further distinguishes penalized spline estimators. Originally, \citet{O:1986} proposed a derivative-based penalty, effectively placing the penalized estimator in between smoothing and regression splines. Later on, \citet{Eilers:1996} recognized the versatility of a discrete difference penalty, giving rise to a completely new class of estimators called P-splines. 

Thanks to their flexible choice of knots and penalties, penalized splines with both derivative and difference-based penalties have in recent years become essential tools in data analysis and constitute the building blocks of many complex estimation methods. However, despite this widespread popularity of penalized-spline estimators, the literature has overwhelmingly focused on the theoretical study of the particular class of least-squares penalized spline estimators. \citet{Li:2008} studied the class of P-spline estimators with lower degree splines combined with a large number of knots and derived the equivalent kernel representation. Their results were subsequently extended to cover splines of arbitrary degree but still with a large number of knots~\citep{Wang:2011}. \citet{Claeskens:2009} provided a theoretical study of least-squares penalized spline estimators with derivative based-penalties, called O-splines, identifying the transition point in their asymptotic properties. More specifically, these authors showed that least-squares penalized O-splines with a small number of knots essentially exhibit regression spline asymptotics while with a large number of knots O-splines essentially behave like smoothing splines. The study of least-squares penalized spline estimators was unified by \citet{Xiao:2019}, who extended the results of \citet{Claeskens:2009} and established asymptotic properties of P-spline estimators in a more general context than \citet{Li:2008} and \citet{Wang:2011}. 

Since the least-squares criterion leads to penalized estimates that are vulnerable to atypical observations and model misspecification, a number of authors have considered alternative methods of estimation. Penalized splines based on more robust loss functions have scarcely appeared in the literature through the years in a number of different contexts, but mostly with little theoretical support. In nonparametric regression, \citet{Lee:2007} proposed replacing the square loss with a more resistant loss function in order to produce a robust O-spline estimator. \citet{T:2010} proposed minimizing a robust scale of the residuals in combination with a derivative-based penalty. \citet{Bol:2006} proposed a P-spline estimator for quantile regression with an $L_1$ penalty and monotonicity constraints, while \citet{An:2014} proposed and theoretically investigated a P-spline estimator for quantile regression in varying coefficient models with a bridge-type penalty. These authors established the consistency of the estimator with a slowly growing number of knots, essentially reproducing one of the possible asymptotic scenarios considered by \citet{Xiao:2019}. A family of robust P-spline estimators in the context of generalized additive models was proposed by \citet{Croux:2012}. More recently, \citet{Kalogridis:2021} showed that the convergence rates of the least-squares estimator with a derivative based penalty can be extended to M-type O-spline estimators using a sufficiently smooth loss function in their objective function. Unfortunately, the smoothness condition excludes popular loss functions such as quantile loss or Huber loss, for example. 

In practice, there may be a number of reasons one may opt for P-splines instead of O-splines. We mention, in particular, that P-splines are extremely easy to set up and allow for very flexible estimates. Both of these facts follow from the special difference penalty that these estimators employ. This penalty can be constructed almost mechanically even when higher order penalties are desired and, what is more, the order of the penalty can be chosen independently from the degree of the spline. The latter is not true when derivative-based penalties are used. A practical consequence is that the user is free to modify these parameters as seen fit in order to achieve the desired degree of smoothness. In view of these important benefits, it is curious that a systematic study of P-spline estimators based on general loss functions has not yet been undertaken.

As an important step in this direction, this paper develops a general asymptotic theory of P-spline estimators based on a wide variety of loss functions, convex and non-convex alike, with weak smoothness requirements, greatly expanding the results of \citep{Li:2008}, \citep{Xiao:2019} and \citep{Kalogridis:2021}. We show that the convergence rate of this class of estimators depend on the rate of growth of the number of knots as well as on the rate of decay of the penalty parameter, illustrating the similarities and also differences compared to penalized splines with derivative-based penalties. Our theory also permits the inclusion of fast-converging auxiliary scale estimates without altering the asymptotic properties of the robust P-spline estimators. Moreover, with minor modifications the methodology developed here for robust P-spline estimators can also be used to weaken the smoothness requirements in \citep{Kalogridis:2021}, thereby greatly extending existing results for robust O-spline estimators. 

The rest of the paper is structured as follows. Section~\ref{sec:2} introduces the family of generalized (M-type) P-spline estimators. We explain the construction of these estimators and draw a useful connection with splines based on derivative penalties. Section~\ref{sec:3} is devoted to the study of the asymptotic behavior of the estimators. We show that under weak assumptions, M-type P-spline estimators enjoy the same rates of convergence as the popular least-squares P-spline estimator, without the need for existence of any moments of the error term for suitably chosen loss functions. These results remain valid when a root-n auxiliary scale estimator is included in the objective function. Section~\ref{sec:4} illustrates via a simulation study the competitive performance of M-type P-spline estimators relative to the least-squares estimator for data with Gaussian errors and shows their superior performance for data with long tailed errors. A real-data application is presented in Section~\ref{sec:5}, while Section~\ref{sec:6} summarizes our conclusions. All proofs are collected in the appendix.

\section{The family of M-type P-spline estimators}
\label{sec:2}

\subsection{B-splines}
A spline is defined as a piecewise polynomial that is smoothly connected at its joints (knots). For any fixed integer $p \geq 1$, let $S_{K}^p$ denote the set of spline functions of order $p$ with knots $0=t_0<t_1 \ldots <t_{K+1}=1$. For $p=1,\ S_{K}^1$ is the set of step functions with jumps at the knots while for $p\geq 2$,
\begin{align*}
S_{K}^p = \{ s \in \mathcal{C}^{p-2}([0,1]): s(x)\ &\text{is a polynomial of degree $(p-1)$ } %\\ & 
\text{on each subinterval $[t_i, t_{i+1} ]$} \}.
\end{align*}
Thus, $p$ controls the smoothness of the functions in $S_{K}^p$ while the number of interior knots $K$ represents the degree of flexibility of spline functions in $S_{K}^p$, see \citep{Ruppert:2003} and \citep{Wood:2017} for insightful discussions in this respect. It is easy to see that $S_{K}^p$ is a $(K+p)$-dimensional subspace of $\mathcal{C}^{p-2}([0,1])$ and a stable basis with good numerical properties is provided by the celebrated B-spline functions, which we now briefly describe~\citep[see][for a full treatment]{DB:2001}.

Let $\{t_k\}_{k=1}^{K+2p}$ be an augmented and relabelled sequence of knots obtained by repeating $t_0$ and $t_{K+1}$ exactly $p$ times. The B-spline functions are defined as linear combinations of truncated polynomials, i.e.,
\begin{equation}
\label{eq:2}
B_{k, p}(x) = \left(t_{k+p}-t_k\right) \left[t_k, \ldots, t_{k+p} \right] (t-x)_{+}^{p-1}, \qquad k=1, \ldots, K+2p,
\end{equation}
where for a function $g$ the placeholder notation $\left[t_i, \ldots, t_{i+p} \right]g$ denotes the $p$th order divided difference of $g(\cdot)$ at $t_i, \ldots, t_{i+p}$, see \citep[pp. 3-10]{DB:2001}. Among other interesting properties B-splines of order $p$ satisfy
\begin{itemize}
\item[(a)] Each $B_{k,p}$ is a polynomial of order $p$ on each interval $(t_k, t_{k+1})$ and has $(p-2)$ continuous derivatives.
\item[(b)] $0<B_{k,p }(x) \leq 1$ for $x \in (t_k, t_{k+p})$ and $B_{k,p}(x)=0$ otherwise. 
\item[(c)] $\sum_{k=1}^{K+2p} B_{k,p}(x) = 1$ for all $x \in [0,1]$.
\end{itemize}
Property (b) is referred to as the \textit{local support} property of the B-spline basis and is the main reason this basis system is so attractive for digital computing and functional approximation. It can be shown that this is the smallest possible support for any basis system for $S_{K}^p$, so that B-splines are in this sense the optimal spline basis. Property (c) is referred to as \textit{partition of unity} and an important implication is that B-spline functions are uniformly bounded. Further properties of splines and the B-spline basis may be found in the classical monographs of  \citet{DB:2001}  and \citet{Schumaker:2007}.

\subsection{P-spline estimators with general loss function}
The idea guiding P-spline estimators is the use of a rich spline basis defined for simplicity on equidistant knots, often with 30 or 40 knots, in order to minimize the approximation bias stemming from the approximation of a generic regression function with a spline, while also penalizing roughness with a difference penalty on the coefficients of the spline. To describe this penalty in more detail, we need to introduce some notation. First, we define the interior knots $t_i = (i-p)/(K+1), \ i =p+1, \ldots, K + p$ and let  $B_{k,p}, \ k =1, \ldots, K+2p$, denote the resulting B-spline basis functions. Then, let $\Delta$ denote the backward difference operator, i.e., $\Delta \beta_j = \beta_j - \beta_{j-1} $ and let $\Delta^q$ denote the composition $\Delta \Delta^{q-1}$. For example, $\Delta^2 \beta_j =\beta_j - 2 \beta_{j-1} + \beta_{j-2}$. Below we examine this penalty in more detail and derive a useful connection with the more intuitive O-spline penalty, but first we describe the M-type P-spline estimator in detail.

With the above notation, the M-type P-spline estimator of $f_0$ is now defined as the spline function $\widehat{f}(x) = \sum_{j=1}^{K+p} \widehat{\beta}_j B_{j,p}(x)$ with $\widehat{\boldsymbol{\beta}}$ the solution of
\begin{align}
\label{eq:3}
\widehat{\boldsymbol{\beta}} = \argmin _{\boldsymbol{\beta} \in \mathbb{R}^{K+p}} \left[ \frac{1}{n} \sum_{i=1}^n \rho\left( Y_i - \sum_{j=1}^{K+p} \beta_j B_{j,p}(x_i) \right) + \lambda \sum_{k=q+1}^{K+p} (\Delta^q \beta_k)^2 \right],
\end{align}
for some $\lambda \geq 0$ controlling the smoothness of the fit and some nonnegative loss function $\rho$ which satisfies $\rho(0) = 0$, where we have omitted the last $p$ B-splines $B_{K+p+1}, \ldots, B_{K+2p}$ as by property (b) their support is outside of $[0,1]$ and thus they do not contribute to the value of the objective function. The loss function $\rho(x) = x^2$ leads to the well-known P-spline estimator first  proposed by \citet{Eilers:1996}. However, the general formulation in~(\ref{eq:3}) permits more general loss functions that reduce the effect of large residuals. A popular example is Huber's function \citep{Huber:1964} given by
\begin{align*}
\rho_{k}(x) = \begin{cases} x^2/2 & |x| \leq k \\ k |x| - k^2/2 & |x| >k,
 \end{cases}
\end{align*}
for some $k >0$ controlling the blending of square and absolute losses. The minimal requirements on the loss function in Section~\ref{sec:3} are also satisfied by many other popular loss functions, such as the absolute loss, Tukey's loss and Hampel's loss~\citep[see, e.g.,][]{Maronna:2006}. Furthermore, since we do not require $\rho$ to be symmetric, our definition also includes P-spline estimators for the conditional quantiles and expectiles of $Y$, for which the theoretical understanding in the literature is rather limited.

For convex $\rho$-functions with a continuous derivative $\psi(x)$, identifying the minimizer of 
$\eqref{eq:3}$ is equivalent to finding $\boldsymbol{\widehat{\beta}}$ such that
\begin{align}
\label{eq:4}
-\frac{1}{n} \sum_{i=1}^n \psi\left( Y_i - 
%\sum_{j=1}^{K+p} \widehat{\beta}_{j} B_{j,p}(x_i) \right) 
\mathbf{B}_{K,p}^{\top}(x_i)\boldsymbol{\widehat{\beta}}\right)
\mathbf{B}_{K,p}(x_i) + 2 \lambda \mathbf{P}_q^{\top} \mathbf{P}_q \boldsymbol{\widehat{\beta}} = \mathbf{0}_{K+p},
\end{align}
where $\mathbf{B}_{K,p}(x)$ is the $(K +p)$-dimensional vector of B-splines evaluated at $x$ and $\boldsymbol{P}_{q}$ is the $(K+p-q) \times (K+p)$ matrix representative of the operator $\Delta^q$ on $\mathbb{R}^{K+p}$. For non-convex $\rho$-functions the minimizer of \eqref{eq:3} satisfies \eqref{eq:4}, but the equivalence between \eqref{eq:3} and \eqref{eq:4} is lost due to the possible existence of local minima. However, in either case \eqref{eq:4} suggests a fast iterative reweighted least squares algorithm for the determination of a solution of the set of the estimating equations, see e.g., \citep{Kalogridis:2021}.

A natural extension is to include a preliminary scale estimate $\widehat{\sigma}$ in \eqref{eq:3} which is easily achieved by modifying the  loss function $\rho$ to $\rho_{\widehat{\sigma}}(x) := \rho(x/\widehat{\sigma})$. The standard approach in robust statistics is to use a scale estimate $\widehat{\sigma}$ computed from the residuals of an initial robust fit to the data  \citep[see, e.g.,][]{Maronna:2006}. While this procedure yields a robust scale estimate, it is computationally demanding and the theoretical properties of such scale estimates are difficult to establish for nonparametric regression. As an alternative, we propose utilizing a robust scale constructed from consecutive differences of the responses, as proposed by \citet{Ghement:2008}. In particular, we use the scale estimator $\widehat{\sigma}$ obtained as the solution of
\begin{align}
\label{eq:5}
\frac{1}{n-1} \sum_{i=1}^{n-1} \rho_c \left( \frac{Y_{i+1}- Y_i}{2^{1/2} \widehat{\sigma}} \right) = \frac{3}{4}.
\end{align}
Here, the loss function $\rho_c$ is the bounded Tukey bi-square given by
\begin{align*}
\rho_c(x) = \begin{cases} 1-\left\{1-(x/c)^2 \right\}^3 & |x| \leq c \\
1 & |x| >c, \end{cases}
\end{align*}
with tuning parameter $c$ equal to $0.704$. The constants $2^{1/2}$ and $3/4$ ensure Fisher-consistency of the scale estimator at the Gaussian distribution and maximal breakdown value, respectively.

\subsection{The P-spline penalty}
In general, the order of the penalty $q$ in \eqref{eq:3} is chosen by the practitioner and reflects how smooth $f_0$ is anticipated to be. However, contrary to the case of O-splines with penalty equal to $\int_{0}^1 |f^{(q)}(x)|^2 dx$ for  $f \in S_K^p$ and $q < p$, the way in which the difference penalty enforces smoothness is not immediately obvious. In general, the difference operator $\Delta^q$ "annihilates" polynomials of order $q$, but this does not necessarily imply anything regarding $\widehat{f}$.  As it turns out, the penalty is intimately linked to the properties of B-spline functions and in order to obtain a better understanding we now derive a useful connection between the O-spline and the P-spline penalties. Specifically, for $f \in S_K^p$ with equidistant knots, \citet[p. 117]{DB:2001} gives the following differentiation formula
\begin{align*}
f^{(q)}(x) = K^q \sum_{j=q+1}^{K+p}  (\Delta^q \beta_j) B_{j,p-q}(x),
\end{align*}
where we have yet again omitted the last $p$ B-splines $B_{K+p+1}, \ldots, B_{K+2p}$ as they are identically equal to zero on $[0,1]$. Squaring and integrating we obtain
\begin{align}
\label{eq:6}
\int_{0}^1 |f^{(q)}(x)|^2 dx  = K^{2q} \sum_{i = q+1}^{K+p} \sum_{j=q+1}^{K+p}  (\Delta^q \beta_j) (\Delta^q \beta_i) \int_{0}^1  B_{j,p-q}(x) B_{i,p-q}(x) dx.
\end{align}
Now, the right-hand side of \eqref{eq:6} is a quadratic form in $\Delta^q \beta_j$ with coefficients given by the entries of the $2(p-q)$-banded matrix $\mathbf{G}_{ij}:= \int_{0}^1 B_{j,p-q}(x) B_{i,p-q}(x) dx$. By Theorem 5.4.2 of \citep{Devore:1993} there exist positive constants $c_1$ and $c_2$ depending only on $p$ and $q$ such that for each spline $\tilde{f} = \sum_{j} \tilde{\beta}_j B_{j,p-q}$,
\begin{align*}
c_1 \sum_{j} \tilde{\beta}_j^2 \leq K \int_{0}^1 |\tilde{f}(x)|^2 dx \leq  c_2 \sum_{j} \tilde{\beta}_j^2.
\end{align*}
From this we may deduce that all the eigenvalues  of $\mathbf{G}$ are in $[c_1 K^{-1},  c_2K^{-1}]$ and it now follows from \eqref{eq:6} that
\begin{align*}
c_1 K^{2q-1} \sum_{j=q+1}^{K+p} (\Delta^q \beta_j)^2  \leq \int_{0}^1 |f^{(q)}(x)|^2 dx \leq  c_2 K^{2q-1} \sum_{j=q+1}^{K+p} (\Delta^q \beta_j)^2.
\end{align*}
We have thus established the following result relating the difference and derivative penalties.
\begin{prop}
\label{prop:1}
If $p>q$, then there exist positive constants $c_2 \geq c_1$ depending only on $p$ and $q$ such that for every $f = \sum_{j} \beta_j B_{j, p}$ we have
\begin{align*}
c_1 \sum_{j=q+1}^{K+p} (\Delta^q \beta_j)^2  \leq K^{1-2q} \int_{0}^1 |f^{(q)}(x)|^2 dx \leq c_2 \sum_{j=q+1}^{K+p} (\Delta^q \beta_j)^2.
\end{align*}
\end{prop}
Proposition \ref{prop:1} shows that the null spaces of derivative-based and difference-based penalties of the same order are identical, thus leading to a better understanding of the difference penalty in \eqref{eq:3}. Indeed, for large $\lambda$ the estimator becomes  a polynomial of degree at most $(q-1)$ while for small $\lambda$ the penalized estimator reduces to a regression spline estimator, which is likely to be very wiggly due to the use of a large number of knots. Combining these two observations reveals that for $\lambda>0$ the difference penalty shrinks the estimator $\widehat{f}_n$ towards a polynomial of degree $(q-1)$. For example, when $q=2$ the difference penalty "pulls" the estimators towards an affine function throughout $[0,1]$.

Although the null spaces of derivative and difference based penalties are identical, that is not to say that the corresponding estimators will be identical in practice. In particular, the behaviour of O-spline and P-spline estimators may differ significantly near the boundaries of the interval $[0,1]$. O-spline estimators can be shown to be  polynomials of order $q$ over $[0, t_{p+1}]$ and $[t_{K+p},1]$, a property that they inherit from smoothing splines. Interestingly, as \citep{Wand:2008} demonstrate, this boundary adjustment may not hold for P-spline estimators; these estimators quite often remain of order $p$ throughout $[0,1]$. It is difficult to predict whether the presence of boundary adjustments has a positive or a negative effect on the penalized estimators. On the one hand, these so-called \textit{natural} boundary conditions reduce the variance of the estimator near the boundary occasionally leading to overall gains \citep{Wand:2008}. On the other hand, the approximation bias is increased \citep{DB:2001} to a degree that it has prompted authors to adapt the integrated penalty near the boundary \citep[see e.g.,][]{Oe:1992}.

\section{Asymptotic properties}
\label{sec:3}
\subsection{M-type P-splines with scale known or redundant}

We now investigate the asymptotic properties of P-spline estimators  based on general loss functions. First we focus on the properties of M-type P-spline estimators where either the scale is known, in which case it can be absorbed into the loss function, or it is not required, as in the case of quantile and expectile P-spline estimators. The assumptions needed for our theoretical development are given in two parts. The first two assumptions concern the design points and knots while the other assumptions concern the loss function.

\begin{itemize}
\item[A.1] Let $Q_n$ denote the empirical distribution of  the design points $x_i \in [0,1], i = 1, \ldots, n$. It is assumed that there exists a distribution function $Q$ with corresponding density $w$ bounded away from zero and infinity such that $\sup_{x}|Q_n(x) - Q(x)| = o(K^{-1})$.
\item[A.2] The number of knots $K = K_n \to \infty$ as $n \to \infty$ and $K=o(n)$.
\end{itemize}
Assumption A.1 is standard in spline estimation, at least going back to \citep{Shen:1998}, and essentially ensures that the design points are well-spread throughout the $[0,1]$-interval. Assumption A.2 is a weak restriction on the rate of growth of the knots as the sample size tends to infinity. Both of these assumptions are also used for least-squares P-spline estimators  \citep{Xiao:2019}. For the broad family of estimators considered herein we will additionally require the following set of assumptions.
\begin{itemize}
\item[A.3] The loss function $\rho(x)$ is absolutely continuous with derivative $\psi(x)$ existing almost everywhere and satisfies $\rho(0)=0$.
\item[A.4] There exist constants $\kappa$ and $M_1$ such that for all $x \in \mathbb{R}$ and $|y| < \kappa$,
\begin{align*}
|\psi(x+y)-\psi(x)| \leq M_1.
\end{align*}
\item[A.5] There exists a constant $M_2$ such that
\begin{align*}
\sup_{|t| \leq h} \mathbb{E}\{ |\psi(\epsilon_1+t) - \psi(\epsilon_1)|^2 \} \leq M_2 |h|,
\end{align*}
as $h \to 0$.
\item[A.6] $\mathbb{E}\{|\psi(\epsilon_1)|^2 \} \leq \tau^2 \in (0, \infty)$, $\mathbb{E}\{\psi(\epsilon_1)\} = 0$ and
\begin{align*}
\mathbb{E}\{\psi(\epsilon_1+t) \} = \xi t + o(t),
\end{align*}
for some $\xi>0$, as $t \to 0$.
\end{itemize}
Our conditions are reminiscent of the conditions in \citep{Bai:1994} and \citep{He:2000} in the context of unpenalized M-estimation and  allow for a wide variety of loss functions, convex and non-convex alike.  Assumption A.4 requires that the score function has locally uniform bounded increments while assumption A.5 is a little stronger than mean-square continuity at zero. Finally, assumption A.6 is a basic Fisher-consistency condition, variants of which have been widely used in robust estimation  \citep[see, e.g.,][for important examples]{Maronna:2006}. It can  easily be checked that the first part holds if $\psi$ is bounded and odd and the error has a symmetric distribution about zero, for example. The second part requires that the function $m(t):= \mathbb{E}\{\psi(\epsilon_1+t) \}$ is differentiable at zero with strictly positive derivative, denoted here by $\xi$, which is a necessary condition for a consistent local minimum to exist in the limit. This is not a stringent condition and covers many interesting estimators, as we now show.

\begin{example}[Squared loss] In this case $\psi(x) = 2x$ and the second part of assumption A.6 holds with $\xi = 2$, provided that $\mathbb{E}\{\epsilon_{1} \} = 0$,  as in \citet{Xiao:2019}.
\end{example}

\begin{example}[Smooth loss functions] All monotone everywhere differentiable $\psi$ functions with bounded second derivative $\psi^{\prime \prime}(x)$, such as $\rho(x) = \log(\cosh(x))$, satisfy the second part of A.6 if
\begin{align*}
0<\mathbb{E}\{\psi^{\prime}(\epsilon_{1})\} < \infty,
\end{align*}
as in \citet{Kalogridis:2021}.
\end{example}

\begin{example}[Check loss]
First, consider the absolute loss for which $\psi(x)=\sign(x)$. If $\epsilon_{1}$ has a distribution function $F$ with positive density $f$ on an interval about zero, then
\begin{equation*}
\mathbb{E}\{ \sign(\epsilon_{1}+t) \} = 2f(0)t + o(t), \quad \text{as} \quad t \to 0,
\end{equation*}
so that A.6 holds with $\xi = 2f(0)$. This easily generalizes  to the check loss $\rho_{\alpha}(x) = x(\alpha - \mathcal{I}(x<0)), \ \alpha \in (0,1)$, provided that in this case one views the regression function $f_0$ as the $\alpha$-quantile function, that is, $\Pr(Y_{i} \leq f_0(x_i) ) = \alpha$, see \citet{Koenker:2005}.
\end{example} 

\begin{example}[Huber loss]
Now $\psi_k(x) = \max(-k, \min(x,k))$ for some $k>0$. Assuming that $F$ has a positive density in a neighbourhood of $-k$ and $k$ we have that
\begin{equation*}
\mathbb{E}\{ \psi_k(\epsilon_{1}+t) \} = \{F(k)-F(-k)\}t + o(t), \quad \text{as} \quad t \to 0.
\end{equation*}
The term in curly brackets is positive for all $k>0$ if, e.g., $F$ is absolutely continuous with density symmetric about zero, whence we can take $\xi = 2F(k)-1$.
\end{example}

\begin{example}[$L_{q}$ loss with $q \in (1,2)$] Here, $\rho_{q}(x) = |x|^{q}$ and  $\psi_q(x) = q |x|^{q-1} \sign(x)$.  If we assume that $F$ is symmetric about zero, $\mathbb{E}\{| \epsilon_{1}|^{q-1}\} < \infty$ and $\mathbb{E}\{|\epsilon_{1}|^{q-2} \}<\infty$, then
\begin{equation*}
\mathbb{E}\{\psi_q(\epsilon_{1}+t) \} = q(q-1)\mathbb{E}\{|\epsilon_{1}|^{q-2} \} t + o(t), \quad \text{as} \quad t \to 0,
\end{equation*}
see \citet{Arcones:2001}. The latter expectation is finite, if, e.g., $F$ possesses a Lebesgue density $f$ that is bounded at an interval about zero. In this case A.6 holds with $\xi = q(q-1)\mathbb{E}\{|\epsilon_{1}|^{q-2} \}$.
\end{example}

\begin{example}[Expectile loss]
As an alternative to the check loss, consider the expectile loss $\rho_{\alpha}(x) = x^2/2(| \alpha - \mathcal{I}(x\leq 0)|)$ with $\alpha \in (0,1)$, such that $\psi_{\alpha}(x) = (1-\alpha)x \mathcal{I}(x \leq 0) + \alpha x\, \mathcal{I}(x>0)$. Assuming that there is an interval about the origin in which $F$ has no atoms we have
\begin{equation*}
\mathbb{E}\{\psi_{\alpha}(\epsilon_{1}+t) \} =  \{\alpha + (1- 2 \alpha) F(0)\}t + o(t), \quad \text{as} \quad t \to 0.
\end{equation*}
The term in curly brackets is positive for each $\alpha \in (0,1)$. Therefore, A.6 holds with $\xi = \alpha + (1- 2 \alpha) F(0)$.
\end{example}

\begin{example}[Hampel loss]
For positive constants $a\leq b <c<\infty$ consider the non-convex three-point Hampel loss \citep{Hampel:2011} $\rho_{a,b,c}(x)$ given by 
\begin{align*}
\rho_{a,b,c}(x) = \begin{cases} x^2/2 & |x| \leq a \\ 
a(|x|-a/2) & a \leq |x| <b \\
 \frac{a(|x|-c)^2}{2(b-c)} + \frac{a(b+c-a)}{2} & b \leq  |x| < c 
 \\  \frac{a(b+c-a)}{2} & c \leq |x| \end{cases}.
\end{align*}
Then, assuming that $F$ is absolutely continuous and symmetric about zero with Lebesgue-density $f$ we have that
\begin{align*}
\mathbb{E}\{\psi_{a,b,c}(\epsilon_1+t) \} = \left(2F(a)-1 - 2a\frac{F(c)-F(b)}{c-b}\right)t + o(t), \quad \text{as} \quad t \to 0.
\end{align*}
The term in brackets is positive provided, e.g., that $f(x)$ is strictly decreasing in $|x|$. In this case A.6 is satisfied with $\xi = 2F(a)-1 - 2a\frac{F(c)-F(b)}{c-b}$.
\end{example}

It should be noted that although the loss functions in Examples 3--7 are very popular among practitioners, the theoretical properties of the associated P-spline estimators have not been described before. In particular, since the score functions are not twice (not even once) continuously differentiable, none of these estimators is covered by the theory of \citep{Kalogridis:2021}. As the above examples reveal, our conditions permit smoothness to be traded between $\psi$ and $F$ and thus cover a much wider variety of loss functions.

Our aim is to examine convergence of $\widehat{f}$ to $f_0$ with respect to the usual $\mathcal{L}_2$-norm, denoted by $||\cdot||_2$. With the above assumptions we can now state our first theoretical result on the asymptotic properties of M-type P-spline estimators that do not depend on an auxiliary scale estimate, for example, the quantile P-spline. In our asymptotic results both $K$ and $\lambda$ depend on $n$, but for convenience we suppress this dependence in the notation. To lighten the notation further, we also adopt the following abbreviations: $\lambda_K = \lambda K^{1-2q}$, $d_{K,\lambda} = \min\{K, \lambda_K^{-1/2q} \}$ and $s_{K, \lambda} = \min\{\lambda_K^2 K^{2q}, \lambda_K \} $. These quantities also depend on $n$ via $K$ and $\lambda$ but for notational clarity we likewise omit to write $n$ in the subscript.

\begin{theorem}
\label{Thm:1} \sloppy Suppose that assumptions A.1--A.6 hold as well as $\lim_{n} \lambda= \lim_n n^{-1} K d_{K, \lambda} = \lim_n K s_{K, \lambda} =0$ and $\lim_n n^{\delta-1} d_{K,\lambda}^{-1} K^3 = 0$ for some $\delta>0$. If $f_0 \in \mathcal{C}^{j}([0,1]), \ q \leq   j \leq p$, then there exists a sequence $\widehat{f}_n$ of local minimizers of \eqref{eq:3} such that
\begin{equation}
\label{eq:7}
||\widehat{f}_n- f_0||^2_2 =  n^{-1} O_P \left( \min\{K, \lambda_K^{-1/2q} \}  \right) + O_P \left(\min\{ \lambda_K^2 K^{2q}, \lambda_K\} \right) + O_P(K^{-2j}).
\end{equation}
Furthermore, for Lipschitz-continuous $\psi$-functions the condition $\lim_n n^{\delta-1} d_{K,\lambda}^{-1} K^3 = 0$ can be replaced by $\lim_n n^{\delta-1} K^2 = 0$ for some $\delta>0$.
\end{theorem}	
Theorem \ref{Thm:1} establishes the same mean-squared error result as derived by \citet{Xiao:2019} for least-squares loss, for a broad class of estimators under minimal additional assumptions. An important implication is that, unlike the least-squares case, penalized M-estimators with bounded score functions can attain this mean-squared error rate even if the error does not possess any finite moments. The first term on the right-hand side of \eqref{eq:7} corresponds to the variance of the P-spline estimator, whereas the following two terms represent the regularization and modelling bias, respectively. The latter arises from the approximation of a generic $\mathcal{C}^{j}([0,1])$ function by a spline \citep[see][p. 149]{DB:2001}. It is important to note that except for this approximation bias, the error rate simultaneously depends on both $K$ and $\lambda$, highlighting the interplay between the knots and penalty in the asymptotics of penalized spline estimators.

Similarly to the case of O-splines investigated by \citet{Claeskens:2009, Xiao:2019} and \citet{Kalogridis:2021} this error decomposition points to a transition between two asymptotic scenarios, depending on the rate of growth of the knots and the rate of decay of the penalty parameter. In particular, for $f \in \mathcal{C}^p([0,1])$  and $K < \lambda_K^{-1/2q}$ (or equivalently, $\lambda K <1$) for all large $n$, one is led to 
\begin{align*}
||\widehat{f}_n - f_0||^2_2 =  n^{-1}O_P(K) + O_P( \lambda^2 K^{2(1-q)})  + O_P(K^{-2p}),
\end{align*}
which is very similar to the mean-squared error of convex M-type regression spline estimators obtained by setting $\lambda = 0$ \citep{Shi:1995}. In fact, setting $K \asymp n^{1/(2p+1)}$ and $\lambda \asymp n^{-\gamma}$ with $\gamma > (1+p-q)/(2p+1)$ yields $||\widehat{f} - f_0||^2_2 = O_P(n^{-2p/(2p+1)})$, which is the optimal rate of convergence for regression functions in $\mathcal{C}^p([0,1])$  \citep{Stone:1982}.

On the other hand, for $f \in \mathcal{C}^q([0,1])$ and $K \geq \lambda_K^{-1/2q}$ (equivalently, $ \lambda K  \geq 1$) for all large $n$ we obtain what is often referred to as a large number of knots scenario, namely
\begin{align*}
||\widehat{f}_n - f_0||^2_2   = n^{-1} O_P( K^{1-1/(2q)} \lambda^{-1/2q}  ) + O_P(\lambda K^{1-2q})+   O_P(K^{-2q}).
\end{align*}
Here, setting $\lambda \asymp n^{-1/(2q+1)}$ and $K \asymp n^{\beta}$ with $\beta \geq 1/(2q+1)$ leads to $||\widehat{f} - f_0||^2_2 = O_P(n^{-2q/(2q+1)})$, which is the optimal rate of convergence for regression functions in  $\mathcal{C}^{q}([0,1])$. Since $p>q$, in this asymptotic scenario the number of knots grows at a faster rate, justifying the designation "large number of knots scenario". For generic $q \leq  j \leq p$, $K$ and $\lambda$ the P-spline estimator is situated in between these two asymptotic scenarios and is still rate-optimal provided that $K$ and $\lambda$ are selected appropriately.

It is interesting to observe that while these results share some similarities with the corresponding results for O-splines, the transition point between the two asymptotic scenarios is very different. In particular, whereas for O-splines the transition point between the two asymptotic scenarios depends on the magnitude of $\lambda^{-1/2q} K $, for P-splines the quantity of importance is $ \lambda  K$. In practical terms, this means that the penalty parameter of P-splines will, in general, need to be much larger than the penalty parameter for O-splines in order to ensure the same effective degrees of freedom.

\subsection{M-type P-splines with preliminary scale}

We now turn to the problem of P-spline estimators that depend on an auxiliary scale estimate $\widehat{\sigma}$, such as the robust estimator in~(\ref{eq:6}). The scale estimate $\widehat{\sigma}$ often plays the role of a tuning parameter and is very useful for piecewise-defined loss functions, such as the Huber, Hampel and Tukey $\rho$-functions where it controls the size of residuals that should be given lower weight in the estimation. The required assumptions on the loss function and auxiliary scale estimate for the main result of this section are as follows.

\begin{itemize}
\item[B.3] The loss function $\rho(x)$ has a Lipschitz-continuous derivative $\psi(x)$ and for every $\epsilon>0$ there exists $M_{\epsilon}$ such that 
\begin{align*}
|\psi(tx)-\psi(sx)| \leq M_{\epsilon}|t-s|,
\end{align*}
for all $t,s > \epsilon$ and $-\infty<x<\infty$.
\item[B.4] There exists a $\sigma \in (0,\infty)$ such that $n^{1/2}(\widehat{\sigma}-\sigma) = O_P(1)$.
\item[B.5] $\mathbb{E}\{|\psi(\epsilon_1/\sigma)|^2 \} < \infty$, $\mathbb{E}\{\psi(\epsilon_1/\alpha) \} = 0$ for any $\alpha>0$ and
\begin{align*}
\mathbb{E}\left\{\psi\left(\frac{\epsilon_1}{\alpha} + t \right) \right\} = \xi(\alpha)t + o(t), 
\end{align*}
as $t \to 0$, for $\xi(\alpha)$ satisfying $0< \inf_{|\alpha-\sigma|\leq \delta} \xi(\alpha) \leq \sup_{|\alpha-\sigma|\leq \delta} \xi(\alpha) < \infty$ for some $\delta>0$.
\end{itemize}

Assumption B.3 requires that $\rho$ is continuously differentiable and $\psi$ changes slowly in the tail. This also implies that $\psi$ is bounded. This condition is borrowed from \citep{He:1995} and is satisfied, e.g., by common redescending $\psi$-functions and Huber $\psi$-functions. For differentiable $\psi$-functions it suffices that $\sup_{x} |x\psi^{\prime}(x)| < \infty$. The scaling constant in B.4 does not need to be the standard deviation of $\epsilon_1$, since we do not assume that $\epsilon_1$ possesses a second moment. The conditions in B.5 parallel those in A.6, except that we now require the linearisation of $m_{\alpha}(t) := \mathbb{E}\{\psi(\epsilon_1/\alpha + t) \}$ to hold for all $\alpha$ in a neighbourhood of $\sigma$. This assumption can be shown to be satisfied yet again for a wide variety of $\psi$-functions that are not necessarily smooth. For example, in the case of the Huber $\psi$-function, assuming that $F$ is absolutely continuous and symmetric about the origin, we obtain
\begin{align*}
\mathbb{E}\left\{\psi_k\left(\frac{\epsilon_1}{\alpha} + t \right) \right\} = \{2F(k/\alpha) -1\}t + o(t), \quad \text{as} \quad t \to 0,
\end{align*}
so that we may take $\xi(\alpha) = 2F(k/\alpha)-1$ and this is clearly bounded away from zero and infinity for all $k>0$. 

Adopting the notation of Theorem~\ref{Thm:1} we have the following important result.

\begin{theorem}
\label{Thm:2}
Suppose that assumptions A.1--A.2 and assumptions B.3--B.5 hold as well as $\lim_{n} \lambda= \lim_n n^{-1} K d_{K, \lambda} = \lim_n K s_{K, \lambda} =0$ and $\lim_n n^{\delta-1} K^2 = 0$ for some $\delta>0$. If $f_0 \in \mathcal{C}^{j}([0,1]), \ q \leq  j \leq p$, then there exists a sequence $\widehat{f}_n$ of local minimizers of \eqref{eq:3} with loss function $\rho_{\widehat{\sigma}}(x) = \rho(x/\widehat{\sigma})$ such that
\begin{equation*}
||\widehat{f}_n- f_0||^2_2 =  n^{-1} O_P \left( \min\{K, \lambda_K^{-1/2q} \}  \right) + O_P \left(\min\{ \lambda_K^2 K^{2q}, \lambda_K\} \right) + O_P(K^{-2j}).
\end{equation*}
\end{theorem}
The theorem states that standardization with a fast-converging scale estimate does not change the asymptotic properties of M-type P-spline estimators. \citet{Ghement:2008} have shown that the M-scale estimator in \eqref{eq:5} satisfies the root-n assumption in B.4, while at the same time providing good protection against outlying observations. Thus, it provides an effective means of standardization for robust P-spline estimators.

\section{A Monte-Carlo study}
\label{sec:4}

To examine the finite-sample performance of M-type P-spline estimators we consider two representative estimators in this class, namely the convex Huber and the non-convex Tukey P-spline estimators and compare their performance with the popular least-squares P-spline estimator.  The least-squares estimator admits a closed-form solution while general M-type estimators can be computed efficiently through the iteratively reweighted least-squares method proposed by \citet{Kalogridis:2021}. For convex M-estimators the starting value for the algorithm is immaterial as the algorithm can be shown to converge to the solution of \eqref{eq:4} regardless of the starting value. However, for non-convex M-estimators a robust starting value is important as without it the algorithm may converge to a local minimum of \eqref{eq:3}. Thus, for the Huber M-estimate we have used the least-squares estimate as the starting value whereas for the Tukey M-estimate we have relied on the Huber M-estimate for the starting value.

The performance of most non-parametric estimators crucially depends on the smoothing parameter, $\lambda$ in the present work given that $K$ is large but fixed. To select this penalty parameter we have made use of the weighted generalized cross-validation (GCV) criterion
\begin{align*}
\GCV(\lambda) = n^{-1} \sum_{i=1} W_i(\boldsymbol{\widehat{\beta}}_n) \frac{|Y_i - \mathbf{B}_{K,p}^{\top}(x_i) \boldsymbol{\widehat{\beta}}_n |^2}{|1-n^{-1} \Tr \mathbf{H}(\lambda)|^2},
\end{align*}
with $\mathbf{H}(\lambda)$ the pseudo-influence matrix obtained upon convergence of the algorithm and $W_i(\boldsymbol{\beta})$ the weights generated by the estimator. That is, $W_i(\boldsymbol{\beta}) = \psi(r_i(\boldsymbol{\beta}))/r_i(\boldsymbol{\beta})$ with $r_i(\boldsymbol{\beta}) = Y_i - \mathbf{B}_{K,p}^{\top}(x_i) \boldsymbol{\beta}, i = 1, \ldots, n$. We select $\lambda$ as the minimizer of $\GCV(\lambda)$. Throughout the simulation experiments and real-data example to follow we have adopted a
two-step approach in order to identify this minimizer. First, we have determined the approximate location of the minimizer by evaluating GCV($\lambda$) on a grid of $\lambda$ values and afterwards employed a numerical optimizer in the neighborhood of the optimum. Such a hybrid approach is often advisable due to the possible local minima and near-flat regions of the GCV criterion, particularly for non-smooth loss functions.

In our simulation experiments and real-data example we have used  for all estimators a cubic B-spline basis generated by $40$ equidistant knots in the interior of $[0,1]$. These choices result in a rich spline subspace of twice continuously differentiable functions. The order of the penalty, $q$, is set equal to $2$, which is a popular choice among practitioners \citep{Eilers:1996}. For the robust estimators we have standardized the losses using the M-scale given in \eqref{eq:5} and selected values for the tuning parameters that yield $95\%$  efficiency in the location model under Gaussian errors. The least-squares estimator does not require an auxiliary scale estimator. 

To investigate the performance of the estimators we consider the regression model $Y_i = f(x_i) + 0.5\epsilon_i$ where $x_i = i/n$ and $f$ is either of the following functions
\begin{enumerate}
\item $f_1(x) =  \cos(2 \pi x)$,
\item $f_2(x) = 3 \arctan(10(x-0.5))$,
\item $f_3(x) = \phi((x-0.3)/0.1)-\phi((x-0.8)/0.04)$,
\end{enumerate}
where $\phi$ denotes the Gaussian density. All three functions are smooth, but have different shapes, $f_1$ is bowl-shaped, $f_2$ is essentially a sigmoid, while $f_3$ has bumps at $0.3$ and $0.8$. Due to their local characteristics, $f_2$ and $f_3$ are more difficult to estimate than $f_1$.

\begin{table}[H]
\centering
\begin{tabular}{l c c c c c  c c c c}
\hline
& & \multicolumn{2}{c}{LS} & \multicolumn{2}{c}{Huber} & \multicolumn{2}{c}{Tukey} \\ 
$f$ & Error Distribution & Mean & Median & Mean & Median & Mean & Median \\ \cline{1-8} 
\multirow{4}{*}{$f_1$} & Gaussian &  0.032 & 0.025   &  0.031 & 0.024  & 0.032 & 0.024  \\
 & $t_3$ & 0.085 &  0.054  &  0.046  & 0.034   & 0.046  &  0.034  \\
 & $st_{3,0.5}$ & 0.135 &  0.107 & 0.086 & 0.076  & 0.082 & 0.073 \\
 & Mixture Gaussian & 0.321  &  0.197 & 0.053  & 0.040  &  0.043 & 0.034   \\ 
 & Slash & 10458  & 1.661 & 0.165  & 0.129  & 0.145 & 0.103 \\ \cline{1-8}
\multirow{4}{*}{$f_2$} & Gaussian & 0.041  & 0.034   & 0.040  & 0.035  & 0.042 & 0.036  \\
 & $t_3$ &  0.097 & 0.074  & 0.058  & 0.053  &  0.058 & 0.052 \\
 & $st_{3,0.5}$ & 0.150 & 0.127 & 0.101 & 0.093 & 0.098 & 0.092 \\
 & Mixture Gaussian & 0.400  & 0.274  & 0.078  & 0.068 &   0.055 &  0.047  \\ 
 & Slash &  1290 & 2.255 &  0.284  & 0.243 &  0.199  &  0.167  \\ \cline{1-8}
\multirow{4}{*}{$f_3$} & Gaussian & 0.033  & 0.027   & 0.029  & 0.026  &0.030  & 0.026   \\
 & $t_3$ & 0.078 &  0.050 &  0.040 &0.033   & 0.040 &  0.033   \\
 & $st_{3,0.5}$ & 0.126 & 0.106 & 0.082 & 0.073 & 0.077 & 0.070 \\
 & Mixture Gaussian & 0.295 & 0.160 & 0.045  &  0.035  & 0.040 & 0.032    \\ 
 & Slash & 375.3 & 1.274 &  0.127 &  0.076 & 0.110  &   0.069 \\ \cline{1-8}
\end{tabular}
\caption{Means and medians of 1000 MSEs for the least-squares, Huber and Tukey P-spline estimators.}
\label{tab:1}
\end{table}

In order to assess the robustness of the estimators, we have generated the errors according to $5$ different distributions. Next to the standard Gaussian distribution, we also consider a t-distribution with 3 degrees of freedom ($t_3$), a skewed t-distribution with $3$ degrees of freedom and non-centrality parameter equal to $0.5$ ($st_{3,0.5}$), and a mixture of mean-zero Gaussians with standard deviations equal to 1 and 9 and weights equal to 0.85 and 0.15, respectively. Finally, we also used Tukey’s Slash distribution, which is defined as the quotient of a standard Gaussian random variable and an independent standard uniform random variable. To evaluate the performance of an estimator $\widehat{f}_n$ we use the "discretized" mean-squared error given by
\begin{align*}
\MSE = \frac{1}{n} \sum_{i=1}^n |\widehat{f}_n(x_i)-f(x_i)|^2.
\end{align*}
Table~\ref{tab:1} presents the mean and median of the $\MSE$s for 1000 samples of size 60.

The results in Table~\ref{tab:1} confirm the extreme sensitivity of the least-squares estimator to even mild deviations from the Gaussian distribution. In particular, while all three estimators behave roughly the same when the errors follow a Gaussian distribution, the performance of the least-squares estimator markedly deteriorates if the errors follow a slightly more heavy-tailed distribution, such as the $t_3$-distribution. More severe contamination aggravates the problem and as a result the least-squares estimator becomes completely unreliable. On the other hand, the robust estimators maintain a steady overall performance over the error distributions. The extreme Slash distribution somewhat affects the robust estimators, but the effect is very small compared to the least squares estimator. Finally, in the two last more severe contamination scenarios, the non-convex P-spline M-estimator based on Tukey loss clearly outperforms its convex counterpart based on Huber loss. 

\begin{figure}[H]
\centering
\subfloat{\includegraphics[width = 0.495\textwidth]{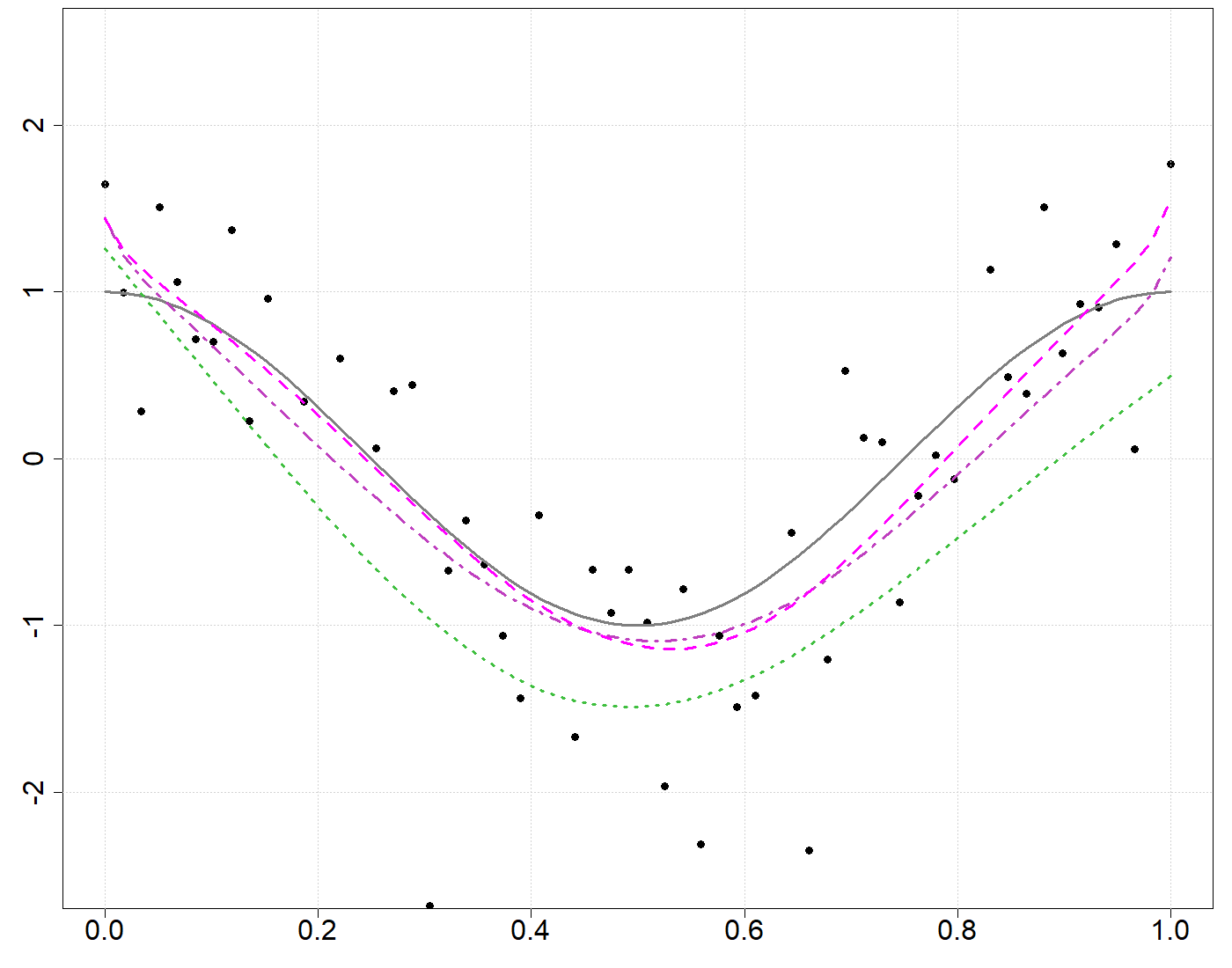}} \ 
\subfloat{\includegraphics[width = 0.495\textwidth]{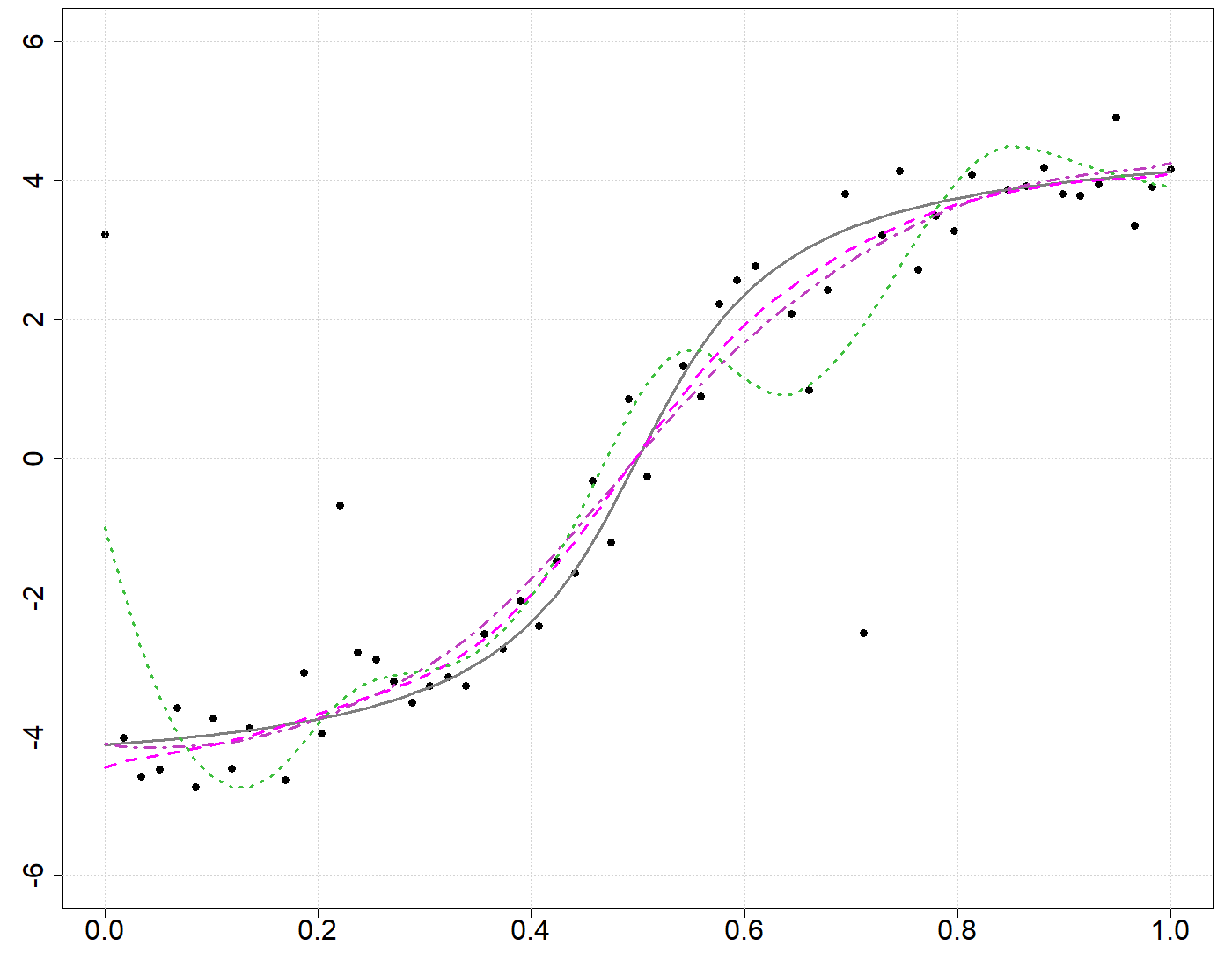}}
\caption{Typical datasets with $f_1$ (left) and $f_2$ (right) as regression functions and the mixture of Gaussians distribution for the error. The lines (\full, \dottedd, \dotdashh, \denselydashed) correspond to the true function, the least-squares, Huber and Tukey P-spline estimators respectively.}
\label{fig:1}
\end{figure}

For a better understanding of the effects of large errors on the least-squares estimator, Figure~\ref{fig:1} presents two typical datasets under the Gaussian mixture error distribution for the first and second regression function, respectively. The plots suggest that the occurrence of large errors has a strong impact not only on the estimated regression function, but also on its smoothness. Under heavy contamination the least-squares estimator tends to either oversmooth or undersmooth, thus concealing essential characteristics of the data. This observation highlights the need not only for robust estimation but also for robust selection of the smoothing parameter (see \citet{Cantoni:2001} for a similar remark).

\section{Application: Historical CO$_2$ emissions in Belgium}
\label{sec:5}

It is well-known that Belgium was one of the first countries in mainland Europe to adopt the new manufacturing processes that characterized the first industrial revolution between late 18th and early 19th centuries. Inevitably, this has led to increased carbon dioxide (CO$_2$) emissions ever since. Figure~\ref{fig:2} presents a scatter plot of the CO$_2$ emissions as a function of the year from 1830, the year of independence of Belgium from the Netherlands, to 2018. Since it is often of interest to both explain and predict the level of emissions, the panel also includes the Tukey, Huber and least-squares P-spline estimators for the overall trend.

The plot suggests that in the aftermath of the industrial evolution CO$_2$ emissions were in a rather steep upward climb that lasted until the late 1970s. Important intermissions in this trend include the two world wars that created such damage to the industry that it took years to return to its prewar output. The expansion of the industry seems to have halted in the 1970s presumably as a result of the twin oil crisis in 1973 that greatly affected the European economies. Interestingly, starting from the 2000s, emissions have been declining and in fact the level of emissions in 2018 matches the level of emissions of the 1960s. In economics it is common to view the two world wars and the two oil crises as temporary shocks. These shocks cause emissions to deviate from their long-run trend so an estimator for the long term pattern should not be overly attracted to these large deviations.

\begin{figure}[H]
\centering
\includegraphics[width = 0.99\textwidth]{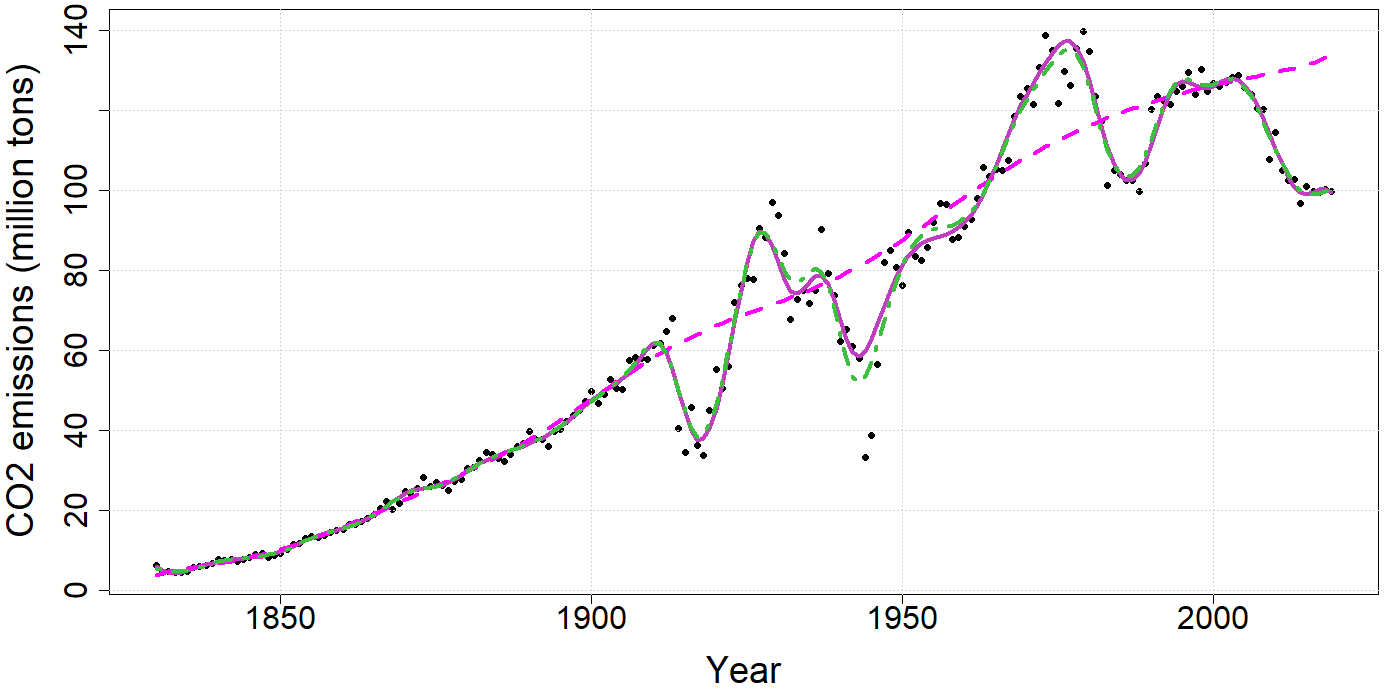}
\caption{Scatter plot of the yearly CO$_2$ emissions versus the year with the Tukey, Huber and LS estimators. The lines (\dotted, \dotdash, \denselydashed) correspond to the least-squares, Huber and Tukey P-spline estimators, respectively. }
\label{fig:2}
\end{figure}

As expected, the least-squares estimator is completely pulled towards the years of unnaturally low level of emissions. What is rather surprising here is that the convex Huber P-spline estimator also demonstrates little resistance to these large deviations, as it produces an estimate that mostly resembles the least-squares estimate. By contrast, the Tukey P-spline estimator exhibits a high degree of resistance, effectively ignoring the shocks and representing the main trend in the data. A likely explanation for this difference is that while the Huber estimator is resistant to a few isolated outliers, it is still susceptible to clusters of outliers. The non-convex Tukey estimator, on the other hand, has a finite rejection point \citep{Hampel:2011} and thus assigns a zero weight to such clusters of outliers resulting in reliable estimates in their presence.

An important benefit of robust estimators is their ability to detect large deviations from the fit by examining the residuals. In particular, since robust estimators are not attracted by outlying observations, these result in large residuals and one can identify them, for example, from a QQ plot of the residuals shown in Figure~\ref{fig:3}. To detect outlying observations, we compare the residuals to a normal distribution. Assuming a normal distribution for the errors corresponding to the majority of regular data is a common approach in robust statistics which often works well in practice to identify large deviations \citep[see e.g.][]{Maronna:2006}.  Clearly, this QQ plot reveals the presence of numerous outlying observations through the past 190 years, most notable of which are the years 1929, 1944, 1945 and 1973. The year 1929 is rather important, since it marks the beginning of the "Great Depression". However, in absence of the residual analysis of a robust estimator, one would only be able to identify it from a very close examination of the scatter plot of the data.

\begin{figure}[H]
\includegraphics[width = 0.99\textwidth]{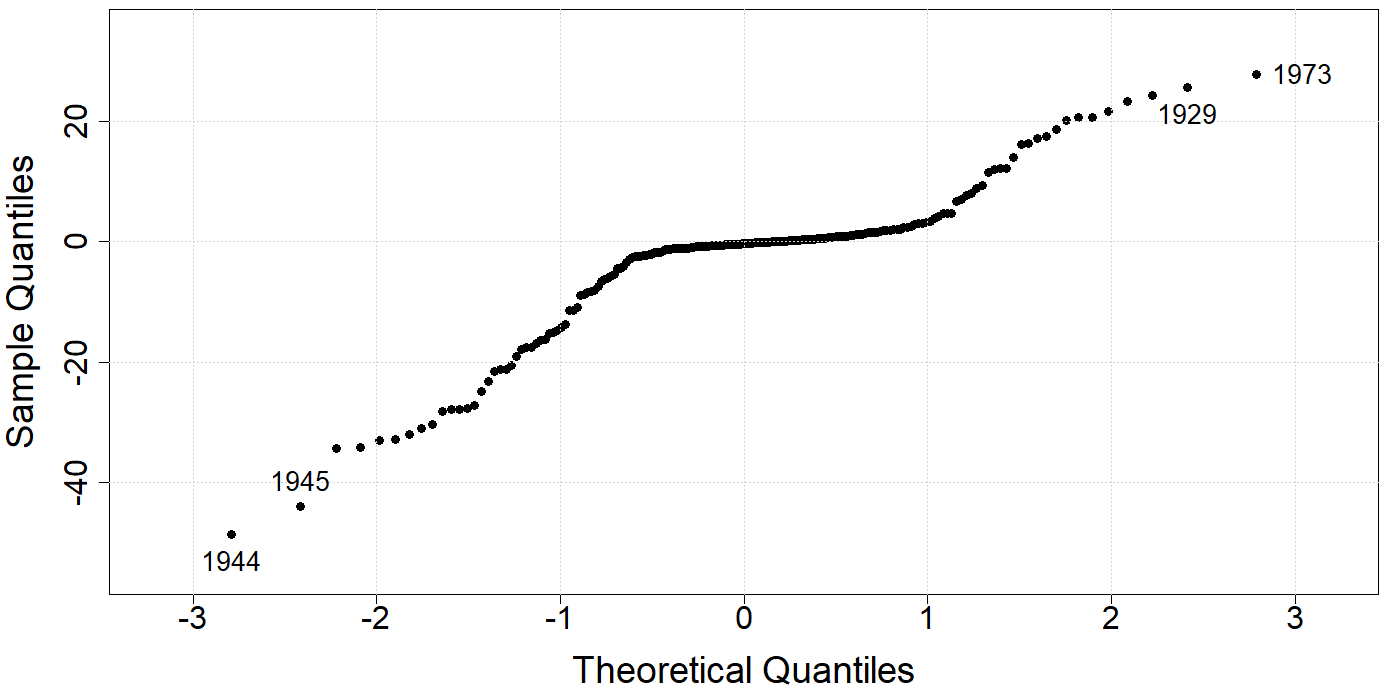}
\caption{Gaussian QQ-plot of the residuals of the Tukey P-spline estimator.}
\label{fig:3}
\end{figure}

\section{Concluding remarks}
\label{sec:6}

The present paper provides theoretical and practical justification for P-spline estimators based on a large class of loss functions. For well-chosen resistant loss functions, only weak assumptions are required to establish the same rates of convergence as for the least squares estimator. Moreover, an appropriate preliminary scale estimate can be included in the loss function which is also useful for outlier detection, as demonstrated in our real-data example. There are several directions worth pursuing from here, of which generalization to random designs and/or higher dimensions seem to come the most natural. To the best of our knowledge, robust penalized estimators in higher dimensions are virtually non-existent in the literature so that such contributions would fill an important void while at the same time providing valuable tools for the applied scientist.

Another important area where robust penalized spline estimators would be successful is functional data analysis, be it in the form of location and dispersion estimation or regression. In the latter case, P-spline estimators constitute natural alternatives to methods based on principal components \citep[see e.g.][]{Kalogridis:2019}. Since, unlike principal components, the B-spline basis is not data-dependent, the use of P-splines in this context would likely lead to more computationally convenient and stable estimators. We aim to explore these interesting directions in future work.

\section*{Acknowledgements}

The authors are grateful to two anonymous referees, the associate editor and the editor for constructive comments and suggestions that lead to a much improved paper in terms of accessibility and content. 
This research was supported by grant C16/15/068 of Internal Funds KU Leuven. Their support is gratefully acknowledged.

\section*{Appendix: Proofs of the theoretical results}

Throughout the appendix we denote $\mathbf{G}_{\lambda} := \mathbf{H} + \lambda \mathbf{P}_q^{\top} \mathbf{P}_q$ with $\mathbf{H} = n^{-1} \sum_{i=1}^n \mathbf{B}_{K,p}(x_i) \mathbf{B}_{K,p}^{\top}(x_i)$ in the notation of Section \ref{sec:2}. To simplify the notation we drop the subscript from the $(K+p)$-dimensional B-spline vector from now on, i.e. we write $\mathbf{B}(x_i)$ instead of $\mathbf{B}_{K,p}(x_i)$. We further use $||\cdot||_{E}$ to denote the Euclidean norm on $\mathbb{R}^{K+p}$, $||\cdot||$ to denote the spectral norm of a square matrix and $||\cdot||_n$ to denote the empirical norm, that is, $||f||_n^2 = n^{-1} \sum_{i=1}^n |f(x_i)|^2$. Generic positive constants are denoted by $c_0$. 

\begin{lemma}
\label{Lem:1}
For each $f \in \mathcal{C}^{j}([0,1])$ there exists a spline function $s_f$ of order $p$ with $p >j$ such that
\begin{equation*}
\sup_{x \in [0,1]} | f(x) - s_f(x) | \leq c_0 |\mathbf{t}|^{j} \sup_{|x-y|<\mathbf{t}} |f^{(j)}(x)-f^{(j)}(y)|,
\end{equation*}
where $t_i$ are the knots,  $\mathbf{t} = \max_i|t_i - t_{i-1}|$ is the maximum distance of adjacent knots and the constant $c_0$ depends only on $p$ and $j$.
\end{lemma}
\begin{proof}
\quad See \citet[pp.145-149]{DB:2001}. 
\end{proof}

\begin{lemma}
\label{Lem:2}
Assume equidistant knots and conditions A.1--A.2. Then there exists a positive constant $c_0$ such that
\begin{equation*}
c_0 K^{-1}\{1+o(1)\} \leq \lambda_{\min}(\mathbf{G}_{ \lambda}),
\end{equation*}
where $\lambda_{\min}(\mathbf{G}_{ \lambda})$ denotes the smallest eigenvalue of $\mathbf{G}_{\lambda}$.
\end{lemma}
\begin{proof}
\quad See Lemma 6.1 of \citet{Shen:1998}.
\end{proof}

\begin{proof}[Proof of Theorem~\ref{Thm:1}]

Let us write $f^{\star}$ for the spline approximation of $f_0$ constructed with the help of Lemma \ref{Lem:1}. Since B-splines form a basis for $S_{K}^p$, $f^{\star} = \sum_{j} \beta_j^{\star} B_j$. Further, let $L_n(\boldsymbol{\beta})$ denote the objective function, that is,
\begin{align*}
L_n(\boldsymbol{\beta}) &= \frac{1}{n} \sum_{i=1}^n \rho(Y_i - \mathbf{B}^{\top}(x_i) \boldsymbol{\beta}) + \lambda \boldsymbol{\beta}^{\top}\mathbf{P}_{q}^{\top} \mathbf{P}_q \boldsymbol{\beta}
 \\ & = \frac{1}{n}\sum_{i=1}^n \rho\left(\epsilon_i + R_i+ \mathbf{B}^{\top}(x_i)(\boldsymbol{\beta}^{\star}- \boldsymbol{\beta}) \right) + \lambda \boldsymbol{\beta}^{\top}\mathbf{P}_{q}^{\top} \mathbf{P}_q \boldsymbol{\beta},
\end{align*}
where $R_i = f_0(x_i) - f^{\star}(x_i), \  i=1, \ldots, n$. Since, by Lemma~\ref{sec:2}, $\mathbf{G}_{\lambda}$ is non-singular for all large $n$ we may reparametrize by setting $\boldsymbol{\gamma} = \mathbf{G}_{\lambda}^{1/2} (\boldsymbol{\beta}^{\star} - \boldsymbol{\beta})$ so that the objective function may be equivalently written as
\begin{align*}
L_n(\boldsymbol{\gamma} ) &  = \frac{1}{n} \sum_{i=1}^n \rho\left(\epsilon_{i} + R_i + \mathbf{B}^{\top}(x_i) \mathbf{G}_{\lambda}^{-1/2} \boldsymbol{\gamma} \right) + \lambda \boldsymbol{\gamma}^{\top} \mathbf{G}_{\lambda}^{-1/2} \mathbf{P}_{q}^{\top} \mathbf{P}_q \mathbf{G}_{\lambda}^{-1/2} \boldsymbol{\gamma}	 \\ &\phantom{{}=1} + \lambda \boldsymbol{\beta}^{\star \top} \mathbf{P}_{q}^{\top}  \mathbf{P}_{q} \boldsymbol{\beta}^{\star} -2 \lambda \boldsymbol{\gamma}^{\top} \mathbf{G}_{\lambda}^{-1/2} \mathbf{P}_{q}^{\top}  \mathbf{P}_{q} \boldsymbol{\beta}^{\star}.
\end{align*}

Since maximizing $L_n(\boldsymbol{\gamma})$ is equivalent to maximizing $L_n(\boldsymbol{\beta})$, we will show that for every $\epsilon>0$ there exists a $D = D_{\epsilon} \geq 1$ such that 
\begin{equation}
\label{eq:8}
\lim_{n \to \infty} \Pr\left( \inf_{||\boldsymbol{\gamma}||_{E} = D} L_n(C_n^{1/2}\boldsymbol{\gamma}) > L(\boldsymbol{0})  \right) \geq 1-\epsilon,
\end{equation}
where $C_n =  n^{-1}\min(K, \lambda_K^{-1/2q}) + \min(\lambda_K^2 K^{2q}, \lambda_K) + K^{-2j}$. This result then implies that for all large $n$ there exists a local minimizer $\widehat{\boldsymbol{\gamma}}$ in the ball $\{\boldsymbol{\gamma} \in \mathbb{R}^{K+p}: ||\boldsymbol{\gamma}||_{E} \leq D C_n^{1/2} \}$, with probability at least $1-\epsilon$. Using the one-to-one relation between $\boldsymbol{\gamma}$ and $\boldsymbol{\beta}$ we further obtain
\begin{align*}
||\widehat{f}-f^{\star}||_n^2 & \leq ||\widehat{f}-f^{\star}||_n^2 + \lambda(\widehat{\boldsymbol{\beta}} - \boldsymbol{\beta}^{\star})^{\top} \mathbf{P}_q^{\top} \mathbf{P}_q(\widehat{\boldsymbol{\beta}} - \boldsymbol{\beta}^{\star}) \\ 
& = ||\mathbf{G}_{\lambda}^{1/2}(\widehat{\boldsymbol{\beta}}-\boldsymbol{\beta}^{\star})||_{E}^2 
\\ & = O_P(C_n).
\end{align*}
From the triangle inequality and the spline approximation property given in Lemma \ref{Lem:1} we consequently obtain 
\begin{align*}
||\widehat{f}-f_0||_n & \leq  ||\widehat{f}-f^{\star}||_n + ||f^{\star}-f_0||_n \\
 & = O_P(C_n^{1/2}) + O(K_n^{-j})
\\ & = O_P(C_n^{1/2}),
\end{align*}
which is the almost the result of Theorem~\ref{Thm:1}. To pass from the empirical norm $||\cdot||_n$ to the $\mathcal{L}_2$-norm $||\cdot||_2$ one can argue as in Corollary 1 of \citet{Kalogridis:2021}.

To establish the theorem it thus suffices to prove \eqref{eq:8}. To that end, use A.3 to decompose $L_n(C_n^{1/2}\boldsymbol{\gamma}) - L_n(\boldsymbol{0})$ as follows
\begin{align*}
L_n(C_n^{1/2}\boldsymbol{\gamma}) - L_n(\mathbf{0}) & =  \frac{1}{n} \sum_{i=1}^n \rho(\epsilon_i  + R_i + C_n^{1/2}\mathbf{B}^{\top}(x_i) \mathbf{G}_{\lambda}^{-1/2} \boldsymbol{\gamma}) 
\\ &\phantom{{}=1} \quad - \frac{1}{n} \sum_{i=1} ^n \rho(\epsilon_i + R_i)  + \lambda C_n \boldsymbol{\gamma}^{\top} \mathbf{G}_{\lambda}^{-1/2} \mathbf{P}_{q}^{\top}  \mathbf{P}_{q} \mathbf{G}_{\lambda}^{-1/2} \boldsymbol{\gamma} 
\\ &\phantom{{}=1} \quad -2 \lambda C_n^{1/2} \boldsymbol{\gamma}^{\top} \mathbf{G}_{\lambda}^{-1/2} \mathbf{P}_{q}^{\top}  \mathbf{P}_{q} \boldsymbol{\beta}^{\star}
\\ & = \frac{1}{n} \sum_{i=1}^n \int_{R_i}^{R_i + C_n^{1/2}\mathbf{B}^{\top}(x_i) \mathbf{G}_{ \lambda}^{-1/2} \boldsymbol{\gamma}} \{ \psi(\epsilon_i+u) - \psi(\epsilon_i) \} du \\   &\phantom{{}=1} \quad + \frac{C_n^{1/2}}{n} \sum_{i=1}^n \mathbf{B}^{\top}(x_i) \mathbf{G}_{ \lambda}^{-1/2} \boldsymbol{\gamma} \psi(\epsilon_i) 
\\  &\phantom{{}=1} \quad+ \lambda C_n \boldsymbol{\gamma}^{\top} \mathbf{G}_{\lambda}^{-1/2} \mathbf{P}_{q}^{\top}  \mathbf{P}_{q} \mathbf{G}_{\lambda}^{-1/2} \boldsymbol{\gamma}
\\ &\phantom{{}=1} \quad -2 \lambda C_n^{1/2} \boldsymbol{\gamma}^{\top} \mathbf{G}_{\lambda}^{-1/2} \mathbf{P}_{q}^{\top}  \mathbf{P}_{q} \boldsymbol{\beta}^{\star}
\\ & := I_1(\boldsymbol{\gamma})  + I_2(\boldsymbol{\gamma})  + I_3(\boldsymbol{\gamma}) ,
\end{align*}
with
\begin{align*}
I_1(\boldsymbol{\gamma}) & := \frac{1}{n} \sum_{i=1}^n \int_{R_i}^{R_i + C_n^{1/2} \mathbf{B}^{\top}(x_i) \mathbf{G}_{\lambda}^{-1/2} \boldsymbol{\gamma}} \{ \psi(\epsilon_i+u) - \psi(\epsilon_i) \} du \\ & \phantom{{}=1}+ \lambda C_n \boldsymbol{\gamma}^{\top} \mathbf{G}_{ \lambda}^{-1/2} \mathbf{P}_{q}^{\top} \mathbf{P}_q \ \mathbf{G}_{ \lambda}^{-1/2} \boldsymbol{\gamma},
\\ I_2(\boldsymbol{\gamma})  & := \frac{C_n^{1/2}}{n} \sum_{i=1}^n \mathbf{B}^{\top}(x_i) \mathbf{G}_{\lambda}^{-1/2}  \boldsymbol{\gamma} \psi(\epsilon_i)
\end{align*}
and
\begin{align*}
I_3(\boldsymbol{\gamma}) & := -2 C_n^{1/2} \boldsymbol{\gamma}^{\top} \mathbf{G}_{\lambda}^{-1/2} \mathbf{P}_{q}^{\top} \mathbf{P}_q \boldsymbol{\beta}^{\star}.
\end{align*}
By the superadditivity of the infimum we have the lower bound
\begin{align*}
 \inf_{||\boldsymbol{\gamma}||_{E} = D} \left[ L_n(C_n^{1/2}\boldsymbol{\gamma}) -  L(\boldsymbol{0}) \right] & \geq  \inf_{||\boldsymbol{\gamma}||_{E} = D} I_1(\boldsymbol{\gamma}) + \inf_{||\boldsymbol{\gamma}||_{E} = D} I_2(\boldsymbol{\gamma})  + \inf_{||\boldsymbol{\gamma}||_{E} = D} I_3(\boldsymbol{\gamma}).
\end{align*}
We determine the order of each term appearing on the right-hand side of the above inequality. Starting with $I_3(\boldsymbol{\gamma})$, observe that for every bounded function $f$ on a set $B$ we have $|\inf_{x \in B} f(x)| \leq \sup_{x \in B} |f(x)|$. Using this, the Schwarz inequality immediately gives
\begin{align*}
\sup_{||\boldsymbol{\gamma}||_{E} \leq D} | I_3(\boldsymbol{\gamma})  | & =  2 \lambda C_n^{1/2} \sup_{||\boldsymbol{\gamma}||_{E} \leq  D} | \boldsymbol{\gamma}^{\top} \mathbf{G}_{ \lambda}^{-1/2} \mathbf{P}_{q}^{\top} \mathbf{P}_q \boldsymbol{\beta}^{\star} | 
\\ & \leq 2 D C_n^{1/2} \lambda || \mathbf{G}_{\lambda}^{-1/2} ( \mathbf{P}_{q}^{\top} \mathbf{P}_q)^{1/2} || \times  || (\mathbf{P}_{q}^{\top} \mathbf{P}_q)^{1/2} \boldsymbol{\beta}^{\star} ||_{E}.
\end{align*}
By Proposition \ref{prop:1},
\begin{align*}
|| (\mathbf{P}_q^{\top} \mathbf{P}_q)^{1/2} \boldsymbol{\beta}^{\star} ||_{E}^2 = \sum_{j=q+1}^{K+p}(\Delta^q \beta_j^{\star} )^2 \leq c_0 K^{1-2q} \int_{0}^1 |f^{\star (q) }(x)|^{2} dx = O(K^{1-2q}),
\end{align*}
as $\int_{0}^1 |f^{\star (q) }(x)|^{2} dx$ is finite for every $f \in \mathcal{C}^{j}([0,1])$ with $j \geq q$, see Theorem (26) of \citet[Chapter XII]{DB:2001}. Next, by the submultiplicativity of the spectral norm, Lemma~\ref{Lem:2} and proposition 4.1 of \citet{Xiao:2019} we obtain
\begin{align*}
|| \mathbf{G}_{\lambda}^{-1/2} (\mathbf{P}_{q}^{\top} \mathbf{P}_q )^{1/2} || \leq || \mathbf{G}_{\lambda}^{-1/2}|| \times  || (\mathbf{P}_{q}^{\top} \mathbf{P}_q )^{1/2}|| = O(K^{1/2}). 
\end{align*}
Moreover, by definition of $\mathbf{G}_{\lambda}$,
\begin{align*}
|| \lambda^{1/2} (\mathbf{P}_{q}^{\top} \mathbf{P}_q )^{1/2} \mathbf{G}_{\lambda}^{-1/2} || \leq 1.
\end{align*}
By combining these bounds we find that
\begin{align}
\label{eq:9}
\sup_{||\boldsymbol{\gamma}||_{E} \leq D} \left| I_3(\boldsymbol{\gamma}) \right| & \leq 2 D C_n^{1/2} \lambda K^{1/2-q}  || \mathbf{G}_{\lambda}^{-1/2} (\mathbf{P}_{q}^{\top} \mathbf{P}_q )^{1/2}|| \nonumber
\\ &  \leq 2 D C_n^{1/2} \lambda^{1/2} K^{1/2-q}  \min\{\lambda^{1/2} K^{1/2}, 1 \} \nonumber
\\ & =  D C_n^{1/2} O\left( \min\{ \lambda K^{1-q}, \lambda^{1/2} K^{1/2-q} \} \right) \nonumber
\\ & = D O(C_n),
\end{align}
where the last line follows from the definition of $C_n$.

Turning to $I_2(\boldsymbol{\gamma})$, an application of the Schwarz inequality yields
\begin{align*}
\sup_{||\boldsymbol{\gamma}||_{E} \leq D} |I_2(\boldsymbol{\gamma})| \leq D C_n^{1/2}  \left\| \frac{1}{n} \sum_{i=1}^n   \mathbf{B}^{\top}(x_i) \mathbf{G}_{\lambda}^{-1/2} \psi(\epsilon_i) \right\|_{E}.
\end{align*}
To bound this term, note that the errors are i.i.d. and thus, by assumption A.6,
\begin{align*}
\mathbb{E} \left\{ \left\| \frac{1}{n} \sum_{i=1}^n   \mathbf{B}^{\top}(x_i) \mathbf{G}_{\lambda}^{-1/2} \psi(\epsilon_i)  \right\|_{E}^2 \right\} & =  \frac{\tau^2}{n} \Tr \{ \mathbf{H}^{1/2} \mathbf{G}_{ \lambda}^{-1} \mathbf{H}^{1/2} \}
\\ & =  \frac{\tau^2}{n} \Tr \{ (\mathbf{I}_{K+p} + \lambda \widetilde{\mathbf{D}}_q)^{-1} \}
\\ & =  \frac{\tau^2}{n} \sum_{j=1}^{K+p} \frac{1}{1+ \lambda \widetilde{s}_j	} ,
\end{align*}
where
\begin{equation*}
\widetilde{\mathbf{D}}_q := \mathbf{H}^{-1/2} \mathbf{P}_q^{\top} \mathbf{P}_q \mathbf{H}^{-1/2},
\end{equation*}
and $\widetilde{s}_j, \ j=1, \ldots, K+p$ are the eigenvalues of $\widetilde{\mathbf{D}}_q$. Under A.1, Lemma 5.1 of \citet{Xiao:2019} implies the existence of positive constants $c_1<c_2$ such that 
\begin{align*}
\widetilde{s}_1 = \ldots= \widetilde{s}_q = 0, \  c_1  (j-q)^{2q} \leq  K^{2q-1}\widetilde{s_j} \leq c_2 j^{2q}, \quad (j=q+1, \ldots, K+p),
\end{align*}
for all large $n$. By integral approximation,
\begin{align*}
\sum_{j=1}^{K+p} \frac{1}{1+ \lambda \widetilde{s}_j} & \leq q + \int_{q}^{K+p} \frac{dx}{1+\lambda K^{1-2q} c_1 (x-q)^{2q}}  \nonumber
\\ & = q +  \int_{0}^{K+p-q} \frac{dx}{1+ \lambda  K^{1-2q} c_1 x^{2q}} \nonumber
\\ & \leq  q + c_1^{-1/2q} \lambda^{-1/2q} K^{1-1/2q} \int_{0}^{K_{q}} \frac{dx}{1+x^{2q}},
\end{align*}
where $K_{q} = c_1^{1/2q} (K+p-q) \lambda^{1/2q} K^{1/2q-1} \asymp K^{1/2q} \lambda^{1/2q}$. The integrand $(1+x^{2q})^{-1}$ is bounded by $1$ for all $x \in \mathbb{R}_{+}$. Consequently,
\begin{align*}
\lambda^{-1/2q} K^{1-1/2q} \int_{0}^{K_{q}} \frac{dx}{1+x^{2q}} \leq \lambda^{-1/2q} K^{1-1/2q} K_{q} = O(K).
\end{align*}
At the same time for all $q \geq 1$ we have
\begin{align*}
\int_{0}^{K_{q}} \frac{dx}{1+x^{2q}}  \leq  \int_{0}^{\infty} \frac{dx}{1+x^{2q}} < \infty,
\end{align*}
which yields
\begin{align*}
\lambda^{-1/2q} K^{1-1/2q} \int_{0}^{K_{q}} \frac{dx}{1+x^{2q}} = O_P(\lambda^{-1/2q} K^{1-1/2q}).
\end{align*}
Combining these two bounds we obtain
\begin{align*}
\mathbb{E} \left\{ \left\| \frac{1}{n} \sum_{i=1}^n   \mathbf{B}^{\top}(x_i) \mathbf{G}_{\lambda}^{-1/2} \psi(\epsilon_i)  \right\|_{E}^2 \right\} = n^{-1} O \left( \min\{K, \lambda_K^{-1/2q} \} \right),
\end{align*}
and consequently, by Markov's inequality,
\begin{align}
\label{eq:10}
\sup_{||\boldsymbol{\gamma}||_{E} \leq D} |I_2(\boldsymbol{\gamma})| =  D \frac{C_n^{1/2}}{ n^{1/2}} O_P \left( \min\{K^{1/2}, \lambda_K^{-1/4q} \} \right)= D C_n O_P(1),
\end{align} 
again by definition of $C_n$.

To determine the order of $I_1(\boldsymbol{\gamma})$ we further decompose this term into
\begin{align}
\label{eq:11}
I_1(\boldsymbol{\gamma}) = \mathbb{E}\{I_1(\boldsymbol{\gamma})\} + [I_1(\boldsymbol{\gamma})-\mathbb{E}\{I_1(\boldsymbol{\gamma})\} ].
\end{align}
We first determine a lower bound for $\inf_{||\boldsymbol{\gamma}||_{E} = D}\mathbb{E}\{I_1(\boldsymbol{\gamma})\}$. As a first step, observe that by Lemma~\ref{Lem:1} and Lemma~\ref{Lem:2},
\begin{align*}
\max_{1\leq i \leq n} \{ |R_i| + C_n^{1/2} |\mathbf{B}^{\top}(x_i) \mathbf{G}_{n \lambda}^{-1/2} \boldsymbol{\gamma}| \} & \leq O(K^{-j}) + D C_n^{1/2} \max_{1\leq i \leq n} ||\mathbf{B}^{\top}(x_i) \mathbf{G}_{\lambda}^{-1/2} ||
\\ & = O(K^{-j}) + D C_n^{1/2} O(K^{1/2})
\\ & = o(1),
\end{align*}
as $n \to \infty$, by our limit assumptions. Consequently, A.6 allows us to write
\begin{align*}
\mathbb{E}\{I_1(\boldsymbol{\gamma}) \} & = \frac{1}{n} \sum_{i=1}^n \int_{R_i}^{R_i + C_n^{1/2} \mathbf{B}^{\top}(x_i) \mathbf{G}_{\lambda}^{-1/2} \boldsymbol{\gamma}} \mathbb{E}\{ \psi(\epsilon_i+u) \} du
%\\ &\phantom{{}=1}  
+ C_n\lambda \boldsymbol{\gamma}^{\top} \mathbf{G}_{\lambda}^{-1/2} \mathbf{P}_{q}^{\top}\mathbf{P}_{q} \mathbf{G}_{n \lambda}^{-1/2} \boldsymbol{\gamma}
 \\ & = \frac{1}{n} \sum_{i=1}^n \int_{R_i}^{R_i + C_n^{1/2} \mathbf{B}^{\top}(x_i) \mathbf{G}_{ \lambda}^{-1/2} \boldsymbol{\gamma}} \{\xi u + o(u)\} du
% \\ &\phantom{{}=1}  
+ C_n\lambda \boldsymbol{\gamma}^{\top} \mathbf{G}_{\lambda}^{-1/2} \mathbf{P}_{q}^{\top}\mathbf{P}_{q}  \mathbf{G}_{\lambda}^{-1/2} \boldsymbol{\gamma}
\\ & =   \frac{1}{2n} \sum_{i = 1}^n \{ C_n \xi |\mathbf{B}^{\top}(x_i) \mathbf{G}_{\lambda}^{-1/2} \boldsymbol{\gamma} |^2 + 2 R_i C_n^{1/2} \mathbf{B}^{\top}(x_i) \mathbf{G}_{\lambda}^{-1/2} \boldsymbol{\gamma}  \} \{ 1+o(1)\}
\\ &\phantom{{}=1}  + C_n\lambda \boldsymbol{\gamma}^{\top} \mathbf{G}_{\lambda}^{-1/2} \mathbf{P}_{q}^{\top}\mathbf{P}_{q}  \mathbf{G}_{\lambda}^{-1/2} \boldsymbol{\gamma}
\\ & := I_{11}(\boldsymbol{\gamma}) + I_{12} (\boldsymbol{\gamma}),
\end{align*}
with
\begin{align*}
I_{11}(\boldsymbol{\gamma}) & := \frac{\xi}{2n} \sum_{i=1}^n C_n |\mathbf{B}^{\top}(x_i)\mathbf{G}_{\lambda}^{-1/2} \boldsymbol{\gamma}|^2 \{1+o(1)\} 
%\\ & \phantom{{}=1}
+  C_n\lambda \boldsymbol{\gamma}^{\top} \mathbf{G}_{\lambda}^{-1/2} \mathbf{P}_{q}^{\top}\mathbf{P}_{q}  \mathbf{G}_{\lambda}^{-1/2} \boldsymbol{\gamma}
\end{align*}
and 
\begin{align*}
I_{12} (\boldsymbol{\gamma}) & := \frac{C_n^{1/2}}{n} \sum_{i=1}^n R_i  \mathbf{B}^{\top}(x_i) \mathbf{G}_{\lambda}^{-1/2} \boldsymbol{\gamma}\{1+o(1) \}.
\end{align*}

Focusing first on $I_{11}(\boldsymbol{\gamma})$, for all $\boldsymbol{\gamma} \in \mathbb{R}^{K+p}$ we have
\begin{align}
\label{eq:12}
I_{11}(\boldsymbol{\gamma}) & \geq \min\left\{ \frac{\xi}{2}, 1 \right\} C_n \boldsymbol{\gamma}^{\top} \left[  \mathbf{G}_{\lambda}^{-1/2} \mathbf{H} \mathbf{G}_{\lambda}^{-1/2} \{1+o(1)\} + \lambda  \mathbf{G}_{\lambda}^{-1/2} \mathbf{P}_q^{\top} \mathbf{P}_q \mathbf{G}_{\lambda}^{-1/2} \right]  \boldsymbol{\gamma} \nonumber
\\ & = c_0 ||\boldsymbol{\gamma}||_{E}^2 C_n\{ 1+o(1)\},
\end{align}
where we have used the definition of $\mathbf{G}_{\lambda}$, i.e.,  $\mathbf{G}_{\lambda}= \mathbf{H} + \lambda \mathbf{P}_q^{\top} \mathbf{P}_q$. 

We next determine an upper bound for $|I_{12}(\boldsymbol{\gamma})|$. By the triangle and Schwarz inequalities, we obtain
\begin{align}
\label{eq:13}
|I_{12}(\boldsymbol{\gamma})| & \leq c_0 C_n^{1/2} \max_{1 \leq i \leq n} |R_i| \frac{1}{n} \sum_{i=1}^n |\mathbf{B}^{\top}(x_i) \mathbf{G}_{\lambda}^{-1/2} \boldsymbol{\gamma}| \nonumber
\\ &  \leq c_0 C_n \left\{  \frac{1}{n} \sum_{i=1}^n |\mathbf{B}^{\top}(x_i) \mathbf{G}_{ \lambda}^{-1/2} \boldsymbol{\gamma}|^2 \right\}^{1/2} \nonumber
\\ & \leq c_0 C_n \left\{ \boldsymbol{\gamma}^{\top} \mathbf{G}_{\lambda}^{-1/2} \mathbf{H} \mathbf{G}_{\lambda}^{-1/2} \boldsymbol{\gamma} \right\}^{1/2} \nonumber
\\ & \leq c_0 C_n ||\boldsymbol{\gamma}||_{E} ,
\end{align}
since $\max_{1\leq i \leq n}|R_i| = O(K^{-j}) = O(C_n^{1/2})$ and $||\mathbf{G}_{\lambda}^{1/2} \mathbf{H} \mathbf{G}_{\lambda}^{1/2}|| \leq 1$, as $\mathbf{P}_{q}^{\top} \mathbf{P}_q$ is positive-semidefinite. Combining \eqref{eq:12} and \eqref{eq:13} we deduce that there exists a strictly positive $c_0$ such that
\begin{align}
\label{eq:14}
\inf_{||\boldsymbol{\gamma}||_{E} = D} \mathbb{E}\{ I_{1}(\boldsymbol{\gamma}) \} & \geq c_0 D^2 C_n\{1+o(1)+o(D^{-1})\},
\end{align}
for all large $n$.

To complete the proof we now use an empirical process argument to show that for every $\delta>0$,
\begin{equation}
\label{eq:15}
\lim_{n \to \infty}\Pr\left( \sup_{||\boldsymbol{\gamma}||_{E} \leq D} \left|I_1(\boldsymbol{\gamma}) - \mathbb{E}\{I_1(\boldsymbol{\gamma}) \} \right| > C_n \delta \right) = 0,
\end{equation}
which then implies 
\begin{align*}
\sup_{||\boldsymbol{\gamma}||_{E} \leq D} I_1(\boldsymbol{\gamma}) - \mathbb{E}\{I_1(\boldsymbol{\gamma}) \} = O_P(1) C_n,
\end{align*}
This in turn, in combination with  \eqref{eq:9}, \eqref{eq:10} and \eqref{eq:11} and \eqref{eq:14} implies that $\inf_{||\boldsymbol{\gamma}|| = D}\mathbb{E}\{I_{1}(\boldsymbol{\gamma})\}$ will be positive and dominate all other terms for sufficiently large $D$. Therefore, \eqref{eq:8} holds. 

To prove \eqref{eq:15}, let us view each pair $(x_i, \epsilon_i), i = 1, \ldots, n$, as a random tuple with distribution $P_i = \delta_{x_i} \times F$, with $\delta_{x_i}$ the Dirac-measure at $x_i$ and  $F$ the distribution of the error. Then $\bar{P} := n^{-1} \sum_{i=1}^n P_i = Q_n \times F$, where in accordance with A.1, $Q_n$ is the empirical distribution of the design points.  Further let $P_n$ denote the empirical measure placing mass $n^{-1}$ on each $(x_i, \epsilon_i)$, i.e., $P_n = n^{-1} \sum_{i=1}^n \delta_{x_i, \epsilon_i}$. Then, adopting the notation of \citet[Chapters 5 and 8]{van de Geer:2000}, we have
\begin{align*}
I_{1}(\boldsymbol{\gamma}) - \mathbb{E}\{I_1(\boldsymbol{\gamma}) \} = \int f_{\boldsymbol{\gamma}} d(P_n-\bar{P}) = n^{-1/2} v_n ( f_{\boldsymbol{\gamma}} ),
\end{align*}
where $v_n(\cdot)$ is the empirical process and $f_{\boldsymbol{\gamma}}$ is the function $[0,1] \times \mathbb{R} \to \mathbb{R}$ defined as 
\begin{align*}
f_{\boldsymbol{\gamma}}(x,y) = \int_{R(x)}^{R(x) + C_n^{1/2} \mathbf{B}^{\top}(x) \mathbf{G}_{\lambda}^{-1/2} \boldsymbol{\gamma} } \{\psi(y+u) - \psi(y) \} du,
\end{align*}
where $R(x) = f_0(x) - f^{\star}(x)$ and $\boldsymbol{\gamma}$ satisfies $||\boldsymbol{\gamma}||_{E} \leq D$. This class of functions depends on $n$ but we suppress this dependence for notational convenience. 
We will apply Theorem 5.11 of \citet{van de Geer:2000} to this empirical process adapted for independent but not identically distributed random variables $(x_i, \epsilon_i)$ \citep[see the remarks in][pp. 131--132]{van de Geer:2000}. To this end we first derive a uniform bound on $f_{\boldsymbol{\gamma}}$ and a bound on its $\mathcal{L}^2 (\bar{P})$-norm. For the former note that by assumption A.4 and a previous argument 
\begin{align}
\label{eq:16}
\sup_{|| \boldsymbol{\gamma}||_{E} \leq D} \sup_{(x,y) \in [0,1] \times \mathbb{R}} |f_{\boldsymbol{\gamma}}(x,y)| \leq c_0 K^{1/2} C_n^{1/2}.
\end{align}
Similarly, appealing to assumption A.5, we find
\begin{align}
\label{eq:17}
\int |f_{\boldsymbol{\gamma}}|^2 d \bar{P} & = n^{-1} \sum_{i=1}^n \mathbb{E} \left\{ \left|  \int_{R_i} ^{R_i + C_n^{1/2} \mathbf{B}^{\top}(x_i) \mathbf{G}_{\lambda}^{-1/2} \boldsymbol{\gamma}} \{ \psi(\epsilon_i+u) - \psi(\epsilon_i) \} du \right|^2  \right\} \nonumber
\\ & \leq c_0 n^{-1} C_n^{1/2}  \sum_{i=1}^n |  \mathbf{B}^{\top}(x_i)\mathbf{G}_{\lambda}^{-1/2} \boldsymbol{\gamma}| \left|\int_{R_i}^{R_i + C_n^{1/2}\mathbf{B}^{\top}(x_i)\mathbf{G}_{ \lambda}^{-1/2} \boldsymbol{\gamma} } |u|du \right| \nonumber
\\ & \leq c_0 n^{-1} C_n^{1/2}  \sum_{i=1}^n |  \mathbf{B}^{\top}(x_i)\mathbf{G}_{\lambda}^{-1/2} \boldsymbol{\gamma}| \left(|R_i|^2 + C_n | \mathbf{B}^{\top}(x_i)\mathbf{G}_{\lambda}^{-1/2} \boldsymbol{\gamma}|^2 \right) \nonumber
\\ & \leq  c_0 n^{-1} C_n^{1/2} \max_{1\leq i \leq n} |  \mathbf{B}^{\top}(x_i)\mathbf{G}_{\lambda}^{-1/2} \boldsymbol{\gamma}| \nonumber
  \sum_{i=1}^n \left( |R_i|^2 + C_n | \mathbf{B}^{\top}(x_i)\mathbf{G}_{\lambda}^{-1/2} \boldsymbol{\gamma}|^2 \right)  \nonumber
\\ & \leq c_0 K^{1/2} C_n^{3/2}.
\end{align}
Let $\rho_{S}(f_{\boldsymbol{\gamma}})$ denote the Bernstein "norm" given by
\begin{align*}
|\rho_S (f_{\boldsymbol{\gamma}})|^2 = 2S^2 \int \left( e^{ |f_{\boldsymbol{\gamma}}|/S } - 1 - |f_{\boldsymbol{\gamma}}|/S \right)  d \bar{P}, \quad S> 0.
\end{align*}
The bounds in \eqref{eq:16} and \eqref{eq:17} reveal that we may take $S = c_0 K^{1/2} C_n^{1/2}$ and $R= c_0 K^{1/4} C_n^{3/4}$ for appropriately chosen $c_0$ in Theorem 5.11 of \citet{van de Geer:2000}. These choices fulfil the conditions since $R^2/S = c_0 C_n$ so that we may take $C_1 = c_0$ in that theorem.  With these choices, from Lemma 5.8 of \citet{van de Geer:2000} we deduce that $|\rho_S (f_{\boldsymbol{\gamma}})|^2 \leq c_0 R^2$. Furthermore, using $\mathcal{H}_{B,S}(\delta, \{f_{\boldsymbol{\gamma}}, ||\boldsymbol{\gamma}||_{E} \leq D\}, \bar{P})$ to denote the $\delta$-generalized entropy with bracketing in the Bernstein norm $\rho_{S}$, modifying the constants when necessary, Lemma 5.10 of \citet{van de Geer:2000} reveals that 
\begin{align*}
\mathcal{H}_{B,S}(\delta, \{f_{\boldsymbol{\gamma}}, ||\boldsymbol{\gamma}||_{E} \leq D\}, \bar{P}) \leq \mathcal{H}_B(c_0 \delta,\{f_{\boldsymbol{\gamma}}, ||\boldsymbol{\gamma}||_{E} \leq D\}, \bar{P}),
\end{align*}
where $\mathcal{H}_B(\delta,\{f_{\boldsymbol{\gamma}}, ||\boldsymbol{\gamma}||_{E} \leq D\}, \bar{P})$ stands for the usual $L_2(\bar{P})$ $\delta$-entropy with bracketing. 
We next derive a bound for $\mathcal{H}_B(\delta,\{f_{\boldsymbol{\gamma}}, ||\boldsymbol{\gamma}||_{E} \leq D\}, \bar{P})$. Observe that for any $(\boldsymbol{\gamma}_1, \boldsymbol{\gamma}_2) \in \mathbb{R}^{K+p} \times \mathbb{R}^{K+p} $ in the $D$-ball we have
\begin{align*}
|f_{\boldsymbol{\gamma}_1}(x,y)-f_{\boldsymbol{\gamma}_2}(x,y)| & = \left| \int_{R(x) + C_n^{1/2} \mathbf{B}^{\top}(x) \mathbf{G}_{\lambda}^{-1/2} \boldsymbol{\gamma}_2}^{R(x) + C_n^{1/2} \mathbf{B}^{\top}(x) \mathbf{G}_{\lambda}^{-1/2} \boldsymbol{\gamma}_1} \{ \psi(y+u) - \psi(y) \} du \right|
 \\ & \leq c_0 K^{1/2} C_n^{1/2} || \boldsymbol{\gamma}_1 - \boldsymbol{\gamma}_2||_{E}.
\end{align*}
It now follows from Theorem 2.7.11 of \citet{VDV:1996} that 
\begin{align*}
\mathcal{H}_B(\delta,\{f_{\boldsymbol{\gamma}}, ||\boldsymbol{\gamma}||_{E} \leq D\}, \bar{P}) \leq \mathcal{H}(c_0 \delta/(K^{1/2} C_n^{1/2}), ||\boldsymbol{\gamma}||_{E} \leq D),
\end{align*}
where $\mathcal{H}(\delta, ||\boldsymbol{\gamma}||_{E} \leq D)$ refers to the $\delta$-entropy of the Euclidean-ball with radius $D$.
By Lemma 2.5 of \citet{van de Geer:2000} we have
\begin{align*}
\mathcal{H}(c_0 \delta/(K^{1/2} C_n^{1/2}), ||\boldsymbol{\gamma}||_{E} \leq D) \leq c_0 (K+p) \log\left(K^{1/2} C_n^{1/2}/\delta +1 \right).
\end{align*}
With this upper bound on $\mathcal{H}_B(\delta,\{f_{\boldsymbol{\gamma}}, ||\boldsymbol{\gamma}||_{E} \leq D\}, \bar{P})$ we may now bound the bracketing integral in Theorem 5.11 of \citet{van de Geer:2000} as follows. 
\begin{align*}
\int_{0}^{R} \mathcal{H}_{B,S}^{1/2}(u, \{f_{\boldsymbol{\gamma}}, ||\boldsymbol{\gamma}||_{E} \leq D\}, \bar{P})du  & = c_0 (K+p)^{1/2} \int_{0}^{R} \log^{1/2}\left( \frac{K^{1/2}C^{1/2}}{u} + 1 \right) du
\\ &  = c_0 (K+p)^{1/2} K^{1/2} C_n^{1/2} \int_{0}^{R/(K^{1/2}C_n^{1/2})} \log^{1/2} \left(\frac{1}{u}+1 \right) du
\\ & \leq c_0 K^{3/4} C_n^{3/4} \log^{1/2}(n),
\end{align*}
for all large $n$. Fix $\epsilon>0$ and take $a = \epsilon n^{1/2} C_n$ in that theorem, then it follows that
\begin{align*}
\int_{0}^{R} \mathcal{H}_{B,S}^{1/2}(u, \{f_{\boldsymbol{\gamma}}, ||\boldsymbol{\gamma}||_{E} \leq D\}, \bar{P})du/(n^{1/2}C_n) = c_0 \frac{K^{3/4} \log^{1/2}(n)}{C_n^{1/4} n^{1/2}},
\end{align*}
and, by assumption, there exists a $\delta>0$ such that $\lim n^{\delta-1} d_{K,\lambda}^{-1} K^3 = 0$. Hence, by definition of $C_n$, we have
\begin{align*}
\int_{0}^{R} \mathcal{H}_{B,S}^{1/2}(u, \{f_{\boldsymbol{\gamma}}, ||\boldsymbol{\gamma}||_{E} \leq D\}, \bar{P})du/(n^{1/2}C_n) = o(1),
\end{align*}
as $n \to \infty$. Thus, Theorem 5.11 of \citet{van de Geer:2000} may be applied to give
\begin{align*}
\Pr\left( \sup_{||\boldsymbol{\gamma}||_{E} \leq D} |I_1(\boldsymbol{\gamma}) - \mathbb{E}\{I_1(\boldsymbol{\gamma}) \}| \geq \delta C_n  \right) & = \Pr\left( \sup_{||\boldsymbol{\gamma}||_{E} \leq D} |v_n(f_{\boldsymbol{\gamma}})| \geq \delta n^{1/2} C_n \right)
\\ & \leq  c_0 \exp \left[ - c_0 \delta^2 \frac{n C_n^{1/2}}{K^{1/2}} \right],
\end{align*}
for every $\delta>0$ and for all large $n$. The exponential tends to zero by our limit assumptions and thus we have established \eqref{eq:15}. From \eqref{eq:10}, \eqref{eq:14} and \eqref{eq:15}, it now follows that for sufficiently large $D$, $\mathbb{E}\{I_1(\boldsymbol{\gamma}) \}$ is positive and dominates all other terms with  with arbitrarily high probability, which completes the proof. 

The last claim of the theorem may be established by using the Lipschitz condition of $\psi$ to obtain tighter bounds on \eqref{eq:16} and \eqref{eq:17} of  order $K C_n$ and $ K C_n^2$, respectively.

\end{proof}

\begin{proof}[Proof of Theorem~\ref{Thm:2}]

Let $L_{n, \widehat{\sigma}}(\boldsymbol{\beta})$ denote the objective function for the present problem, that is,
\begin{equation*}
L_{n, \widehat{\sigma}}(\boldsymbol{\beta}) = \frac{1}{n} \sum_{i=1}^n \rho\left(\frac{\epsilon_i + R_i + \mathbf{B}^{\top}(x_i)(\boldsymbol{\beta}^{\star} - \boldsymbol{\beta})}{\widehat{\sigma}} \right) + \lambda \boldsymbol{\beta}^{\top} \mathbf{P}_{q}^{\top} \mathbf{P}_q \boldsymbol{\beta}.
\end{equation*}
Using the argument in the proof of Theorem~\ref{Thm:1} and the consistency of $\widehat{\sigma}$, it suffices to show that for some $\delta>0$ and every $\epsilon>0$ there exists a  $D=D_{\epsilon} \geq 1$ such that 
\begin{equation}
\label{eq:18}
\lim_{n \to \infty} \Pr\left(\inf_{||\boldsymbol{\gamma}||_{E} = D} L_{n, \widehat{\sigma}}(C_n^{1/2} \boldsymbol{\gamma}) > L_{n, \widehat{\sigma}}(\mathbf{0}), \ |\widehat{\sigma}-\sigma| \leq \delta \right) \geq 1-\epsilon,
\end{equation}
for $C_n = n^{-1} \min\{K, \lambda_K^{-1/{2q}} \} + \min\{\lambda_K^2 K^{2q}, \lambda_K\} + K^{-2j}$. Write $L_{n, \widehat{\sigma}}(C_n^{1/2} \boldsymbol{\gamma}) - L_{n, \widehat{\sigma}}(\mathbf{0}) = I_1(\boldsymbol{\gamma}, \widehat{\sigma}) + I_2(\boldsymbol{\gamma}, \widehat{\sigma}) + I_3(\boldsymbol{\gamma}, \widehat{\sigma})$ with
\begin{align*}
I_1(\boldsymbol{\gamma}, \widehat{\sigma}) & = \frac{1}{n\widehat{\sigma}} \sum_{i=1}^n
 \int_{R_i}^{R_i + C_n^{1/2}\mathbf{B}^{\top}(x_i)\mathbf{G}_{\lambda}^{-1/2} \boldsymbol{\gamma} } \left\{ \psi\left(\frac{\epsilon_i + u}{\widehat{\sigma}}\right) - \psi\left(\frac{\epsilon_i}{\widehat{\sigma}}\right) \right\}du \\ &\phantom{{}=1} +  C_n\lambda \boldsymbol{\gamma}^{\top} \mathbf{G}_{\lambda}^{-1/2} \mathbf{P}_{q}^{\top} \mathbf{P}_q \mathbf{G}_{\lambda}^{-1/2} \boldsymbol{\gamma},\,
  \\ I_2(\boldsymbol{\gamma}, \widehat{\sigma}) & =  -\frac{C_n^{1/2}}{n\widehat{\sigma}} \sum_{i=1}^n \mathbf{B}^{\top}(x_i) \mathbf{G}_{\lambda}^{-1/2} \boldsymbol{\gamma} \psi\left(\frac{\epsilon_i}{\widehat{\sigma}}\right)
\end{align*}
and
\begin{align*}
 I_3(\boldsymbol{\gamma}, \widehat{\sigma}) & = 2 C_n^{1/2}\lambda \boldsymbol{\gamma}^{\top} \mathbf{G}_{\lambda}^{-1/2} \mathbf{P}_{q}^{\top} \mathbf{P}_q \boldsymbol{\beta}^{\star}.
\end{align*} 
It will be shown that for some positive constant $c_0$,
\begin{align}
\label{eq:19}
\inf_{||\boldsymbol{\gamma}||_{E} = D, |\alpha-\sigma| \leq \delta  } \mathbb{E}\{I_1(\boldsymbol{\gamma}, \alpha)\} &\geq c_0 D^2\{1+o(1)+o(D^{-1})\},
\\
\sup_{||\boldsymbol{\gamma}||_{E} \leq D, |\alpha-\sigma| \leq \delta} |I_1(\boldsymbol{\gamma}, \alpha) - \mathbb{E}\{I_1(\boldsymbol{\gamma}, \alpha)\}| & =  o_P(1) C_n, \label{eq:20}
\\ \sup_{||\boldsymbol{\gamma}||_{E} \leq  D} |I_2(\boldsymbol{\gamma}, \widehat{\sigma})| & = O_P(1) D C_n \label{eq:21}
\end{align}
and
\begin{align}
 \sup_{||\boldsymbol{\gamma}||_{E} \leq D, |\alpha-\sigma| \leq \delta } |I_3(\boldsymbol{\gamma}, \alpha)|  = O(1) D C_n \label{eq:22}.
\end{align}
These are sufficient for \eqref{eq:18} to hold for a suitably large $D$. 

Since $I_3$ does not depend on $\widehat{\sigma}$, the argument of Theorem~\ref{Thm:1} immediately shows that its supremum is indeed  $O(1) D C_n$. To treat $I_2(\boldsymbol{\gamma}, \widehat{\sigma})$ we first use the triangle inequality to obtain the bound
\begin{align*}
(\sigma - \delta)|I_2(\boldsymbol{\gamma}, \widehat{\sigma})| & \leq \left| \frac{C_n^{1/2}}{n} \sum_{i=1}^n \mathbf{B}^{\top}(x_i) \mathbf{G}_{\lambda}^{-1/2} \boldsymbol{\gamma} \psi\left(\frac{\epsilon_i}{\sigma}\right) \right| \\ &\phantom{{}=1}  
+ \left| \frac{C_{n}^{1/2}}{n} \sum_{i=1}^n  \mathbf{B}^{\top}(x_i) \mathbf{G}_{\lambda}^{-1/2} \boldsymbol{\gamma} \left\{ \psi\left(\frac{\epsilon_i}{\widehat{\sigma}}\right) - \psi\left(\frac{\epsilon_i}{\sigma}\right) \right\} \right|.
\end{align*}
The supremum of the first term may be treated as in the proof of Theorem~\ref{Thm:1} to yield the order $O_P(1)D C_n$.  For the second term, choose $ \epsilon^{\prime} = (2 \sigma)^{-1}$ and note that for all large $n$, $\widehat{\sigma}^{-1} > \epsilon^{\prime}$ with high probability. Furthermore,  for any $\delta \in (0, \sigma/2)$ we have
\begin{align*}
\left|\frac{1}{\widehat{\sigma}} - \frac{1}{\sigma} \right| = \frac{|\sigma - \widehat{\sigma}|}{\widehat{\sigma} \sigma} \leq 	 \frac{|\sigma - \widehat{\sigma}|}{(\sigma - \delta) \sigma} \leq 2\frac{|\sigma - \widehat{\sigma}|}{\sigma^2},
\end{align*}
with probability tending to one. Therefore, condition B.3 now reveals the existence of $M_{\sigma}>0$ such that for all $\boldsymbol{\gamma}$ with $||\boldsymbol{\gamma}||_{E} \leq D$,
\begin{align*}
\left| \sum_{i=1}^n  \mathbf{B}^{\top}(x_i) \mathbf{G}_{\lambda}^{-1/2} \boldsymbol{\gamma} \left\{ \psi\left(\frac{\epsilon_i}{\widehat{\sigma}}\right) - \psi\left(\frac{\epsilon_i}{\sigma}\right) \right\} \right| &  \leq \frac{2M_{\sigma}}{\sigma^2} \sum_{i=1}^n | \mathbf{B}^{\top}(x_i) \mathbf{G}_{\lambda}^{-1/2} \boldsymbol{\gamma}| |\widehat{\sigma}-\sigma|
\\ & = O_P(1) D \left\{\sum_{i=1}^n || \mathbf{B}^{\top}(x_i) \mathbf{G}_{\lambda}^{-1/2}||_{E}^2  \right\}^{1/2}
\\ & =  O_P(n^{1/2}) D \min\{ K^{1/2}, \lambda_K^{-1/{4q}} \},
\end{align*}
by the Schwarz inequality and the root-n consistency of $\widehat{\sigma}$. Combining the above now yields $\sup_{||\boldsymbol{\gamma}||_{E} \leq  D} |I_2(\boldsymbol{\gamma}, \widehat{\sigma})|  = O_P(1) D C_n$, which is \eqref{eq:21}.

Turning to $\mathbb{E}\{I_{1}(\boldsymbol{\gamma}, \alpha)\}$, using B.5 a derivation as in the proof of Theorem~\ref{Thm:1} yields
\begin{align*}
\inf_{||\boldsymbol{\gamma}||_{E} = D, |\alpha - \sigma| \leq \delta} \mathbb{E} \{I_{1}(\boldsymbol{\gamma}, \alpha)\} \geq c_0 D^2 C_n\{1+o(1)+o(D^{-1}) \},
\end{align*}
where, by the local boundedness of $\xi(\alpha)$, ${\color{red} c}$ is strictly positive. We have thus established \eqref{eq:19}

To conclude the proof we need to show \eqref{eq:20}.  Adopting the notation of Theorem~\ref{Thm:1}
\begin{align*}
I_1(\boldsymbol{\gamma}, \alpha) - \mathbb{E}\{I_1(\boldsymbol{\gamma},\alpha \} = \int f_{\boldsymbol{\gamma}, \alpha} d \bar{P} = n^{-1/2} v_n (f_{\boldsymbol{\gamma}, \alpha}),
\end{align*}
where $v_n(\cdot)$ is the empirical process and the function $f_{\boldsymbol{\gamma}, \alpha}$ is given by
\begin{align*}
f_{\boldsymbol{\gamma}, \alpha}(x,y) = \frac{1}{\alpha} \int_{R(x)}^{R(x) + C_n^{1/2} \mathbf{B}^{\top}(x) \mathbf{G}_{\lambda}^{-1/2} \boldsymbol{\gamma}} \left\{ \psi\left( \frac{y+u}{\alpha} \right) - \psi \left( \frac{y}{\alpha} \right)   \right\} du,
\end{align*}
for $(x,y) \in [0,1] \times \mathbb{R}$, where $||\boldsymbol{\gamma}||_{E} \leq D$ and $\alpha \in [\sigma - \delta, \sigma + \delta]$.
 Let us write $\mathcal{B}_{D} = \{ \boldsymbol{\gamma} \in \mathbb{R}^{K+p}: ||\boldsymbol{\gamma}||_{E} \leq D \}$ and $V_{\sigma} = [\sigma-\delta, \sigma+\delta]$  for convenience. We again aim to apply Theorem 5.11 of \citep{van de Geer:2000} to this empirical process and thus proceed as in the proof of Theorem~\ref{Thm:1}. Since, by our limit assumptions,
\begin{align*}
\lim_{n \to \infty} \sup_{x \in [0,1]} |R(x)| = \lim_{n \to \infty} C_n^{1/2} \sup_{\boldsymbol{\gamma} \in \mathcal{B}_{D}} \sup_{x \in [0,1]} |\mathbf{B}^{\top}(x) \mathbf{G}_{\lambda}^{-1/2} \boldsymbol{\gamma}| = 0,
\end{align*}
the Lipschitz-continuity of $\psi$ (assumption B.3) yields
\begin{align*}
\sup_{\boldsymbol{\gamma} \in \mathcal{B}_{D}, \alpha \in V_{\sigma}} \sup_{x \in [0,1], y \in \mathbb{R}} |f_{\boldsymbol{\gamma}, \alpha}(x,y)| \leq c_0 K C_n.
\end{align*}
Similarly, using the Schwarz inequality,
\begin{align*}
\int |f_{\boldsymbol{\gamma}, \alpha}|^2 d \bar{P} &\leq \frac{1}{n (\sigma-\delta)^2} \sum_{i=1}^n \mathbb{E}\left\{ \left| \int_{R_i}^{R_i + C_n^{1/2} \mathbf{B}^{\top}(x_i) \mathbf{G}_{\lambda}^{-1/2} \boldsymbol{\gamma}} \left\{  \psi\left( \frac{\epsilon_i+u}{\alpha} \right) - \psi \left( \frac{\epsilon_i}{\alpha} \right) \right\} du \right|^2  \right\}
\\ & \leq c_0 n^{-1} \sum_{i=1}^n | C_n^{1/2} \mathbf{B}^{\top}(x_i) \mathbf{G}_{\lambda}^{-1/2} \boldsymbol{\gamma}| \left| \int_{R_i}^{R_i + C_n^{1/2} \mathbf{B}^{\top}(x_i) \mathbf{G}_{\lambda}^{-1/2} \boldsymbol{\gamma}} |u|^2 du \right|
\\ & \leq c_0 C_n^{2} \max_{1\leq i \leq n} | \mathbf{B}^{\top}(x_i) \mathbf{G}_{\lambda}^{-1/2} \boldsymbol{\gamma}|^2  n^{-1} \sum_{i=1}^n |\mathbf{B}^{\top}(x_i) \mathbf{G}_{\lambda}^{-1/2} \boldsymbol{\gamma}|^2
\\ & \leq c_0 K C_n^2.
\end{align*}
It follows that we may take $S = c_0 K C_n$ and $R = c_0 C_n K^{1/2}$ in Lemma 5.8 of \citet{van de Geer:2000}.

We now bound the generalized Bernstein entropy of the class of functions $f_{\boldsymbol{\gamma}, \alpha},   \boldsymbol{\gamma} \in \mathcal{B}_{D}, \alpha \in V_{\sigma}$. Since these functions are uniformly bounded, we find 
\begin{align*}
H_{B,S}(\delta, \{ f_{\boldsymbol{\gamma}, \alpha},   \boldsymbol{\gamma} \in \mathcal{B}_{D}, \alpha \in V_{\sigma} \}, \bar{P}) \leq H_{B}(c_0 \delta, \{ f_{\boldsymbol{\gamma}, \alpha},   \boldsymbol{\gamma} \in \mathcal{B}_{D}, \alpha \in V_{\sigma} \}, \bar{P} ),
\end{align*}
for some $c_0$, where $H_{B}$ refers to the standard entropy with bracketing. Since $\bar{P}$ is a probability measure we further find
\begin{align*}
H_{B}(c_0 \delta, \{ f_{\boldsymbol{\gamma}, \alpha},   \boldsymbol{\gamma} \in \mathcal{B}_{D}, \alpha \in V_{\sigma} \}, \bar{P} ) \leq H_{\infty}( c_0 \delta, \{ f_{\boldsymbol{\gamma}, \alpha},   \boldsymbol{\gamma} \in \mathcal{B}_{D}, \alpha \in V_{\sigma} \}),
\end{align*}
where $H_{\infty}$ refers to the entropy in the supremum norm, see, e.g, Lemma 2.1 of \citet{van de Geer:2000}. Now, by the triangle inequality,
\begin{align}
\label{eq:23}
|f_{\boldsymbol{\gamma}_1, \alpha_1}(x,y)-f_{\boldsymbol{\gamma}_2, \alpha_2}(x,y)| \leq |f_{\boldsymbol{\gamma}_1, \alpha_1}(x,y)-f_{\boldsymbol{\gamma}_2, \alpha_1}(x,y)| + |f_{\boldsymbol{\gamma}_2, \alpha_1}(x,y)-f_{\boldsymbol{\gamma}_2, \alpha_2}(x,y)|,
\end{align}
for all $(\boldsymbol{\gamma}_1, \alpha_1), (\boldsymbol{\gamma}_2, \alpha_2) \in (\mathcal{B}_{D} \times V_{\sigma})^2$. For the first term in the RHS of \eqref{eq:23}, by the boundedness of $\psi$, we have
\begin{align}
\label{eq:24}
|f_{\boldsymbol{\gamma}_1, \alpha_1}(x,y)-f_{\boldsymbol{\gamma}_2, \alpha_1}(x,y)| & \leq \frac{1}{\alpha_1}\left| \int_{R(x) + C_n^{1/2} \mathbf{B}^{\top}(x) \mathbf{G}_{\lambda}^{-1/2} \boldsymbol{\gamma}_1}^{R(x) + C_n^{1/2} \mathbf{B}^{\top}(x) \mathbf{G}_{\lambda}^{-1/2} \boldsymbol{\gamma}_2} \left\{ \psi\left(\frac{y + u}{\alpha_1} \right) - \psi\left( \frac{y}{\alpha_1} \right) \right\} du \right| \nonumber
\\ &\leq  \frac{2}{\sigma - \delta} ||\psi||_{\infty} C_n^{1/2}|\mathbf{B}^{\top}(x) \mathbf{G}_{\lambda}^{-1/2} (\boldsymbol{\gamma}_1 - \gamma_2)| \nonumber
\\ & \leq c_0 K^{1/2} C_n^{1/2} ||\boldsymbol{\gamma}_1 - \boldsymbol{\gamma}_2||_{E}.
\end{align}
For the second term in \eqref{eq:23}, the tail condition in B.3 and the fact that $\boldsymbol{\gamma}_2 \in \mathcal{B}_{D}$ entail
\begin{align}
\label{eq:25}
|f_{\boldsymbol{\gamma}_2, \alpha_1}(x,y)-f_{\boldsymbol{\gamma}_2, \alpha_2}(x,y)| & \leq 2 \frac{M_{\sigma}}{(\sigma-\delta)^3} C_n^{1/2} K^{1/2} |\alpha_1-\alpha_2| + 2 \frac{||\psi||_{\infty}}{(\sigma-\delta)^2} K^{1/2} C_n^{1/2} \nonumber
\\ & \leq c_0 C_n^{1/2} K^{1/2}|\alpha_1-\alpha_2|,
\end{align}
for some $c_0>0$. Combining \eqref{eq:23}--\eqref{eq:25} reveals that the $\delta$-covering number in the supremum norm of  $f_{\boldsymbol{\gamma}, \alpha},   \boldsymbol{\gamma} \in \mathcal{B}_{D}, \alpha \in V_{\sigma}$ may be bounded by the product of the $\delta/(K^{1/2} C_n^{1/2})$-covering numbers of $\mathcal{B}_{D}$ and $V_{\sigma}$, viz,
\begin{align}
\label{eq:26}
\mathcal{N}(\delta, \{ f_{\boldsymbol{\gamma}, \alpha},\ \boldsymbol{\gamma} \in \mathcal{B}_D, \ \alpha \in V_{\sigma} \},|| \cdot ||_{\infty} ) & \leq \mathcal{N} (c_0 \delta /(K^{1/2} C_n^{1/2}) , \mathcal{B}_D )  \times  \mathcal{N}( c_0\delta /(K^{1/2} C_n^{1/2}), V_{\sigma}).
\end{align} 
By Lemma 2.5 of \citet{van de Geer:2000} we may further bound the RHS with
\begin{align*}
\mathcal{N}(\delta, \{ f_{\boldsymbol{\gamma}, \alpha},\ \boldsymbol{\gamma} \in \mathcal{B}_D, \ \alpha \in V_{\sigma} \},|| \cdot ||_{\infty} )&  \leq \left(\frac{c_0 K^{1/2} C_n^{1/2}}{\delta} + 1\right)^{K+p} \times \left(\frac{c_0 K^{1/2} C_n^{1/2}}{\delta} + 1\right) \\ &  \leq \left(\frac{c_0 K^{1/2} C_n^{1/2}}{\delta} + 1\right)^{K+p+1},
\end{align*}
for some finite $c_0$, which implies that
\begin{align}
\label{eq:27}
H_{B,S}(\delta, \{ f_{\boldsymbol{\gamma}, \alpha},   \boldsymbol{\gamma} \in \mathcal{B}_{D}, \alpha \in V_{\sigma} \}, \bar{P}) \leq (K+p+1)\log(c_0 K^{1/2} C_n^{1/2}/\delta + 1),
\end{align}
for all $\delta>0$.

With the upper bound of the Bernstein entropy given in \eqref{eq:27}, the bracketing integral in Theorem 5.11 of \citet{van de Geer:2000} may be bounded by

\begin{align*}
\int_{0}^{R} H_{B,S}^{1/2}(u, \{ f_{\boldsymbol{\gamma}, \alpha},   \boldsymbol{\gamma} \in \mathcal{B}_{D}, \alpha \in V_{\sigma} \}, \bar{P}) du & \leq(K+p+1)^{1/2} \int_{0}^{R}\log^{1/2}(c_0 K^{1/2} C_n^{1/2}/u + 1) du 
\\ & \leq c_0 K C_n^{1/2} \int_{0}^{c_0 R/(K^{1/2} C_n^{1/2})} \log^{1/2}(1/u + 1) du
\\ & \leq c_0 K C_n \log^{1/2}(n),
\end{align*}
as $R = c_0 K^{1/2} C_n$ and, by our limit assumptions, $K \to \infty$ and $C_n \to 0$ as $n\to \infty$. It now follows that

\begin{align*}
\int_{0}^{R} H_{B,S}^{1/2}(u, \{ f_{\boldsymbol{\gamma}, \alpha},   \boldsymbol{\gamma} \in \mathcal{B}_{D}, \alpha \in V_{\sigma} \}, \bar{P}) du/(n^{1/2} C_n) & \leq c_0 K \log^{1/2}(n)/n^{1/2}
\\ & = o(1),
\end{align*}
as $n\to \infty$, by our limit assumptions. Thus, Theorem 5.11 of \citet{van de Geer:2000} may be applied to give

\begin{align*}
\Pr\left( \sup_{\boldsymbol{\gamma} \in \mathcal{B}_{D}, \alpha \in V_{\sigma}} |I_1(\boldsymbol{\gamma}, \alpha) - \mathbb{E}\{I_1(\boldsymbol{\gamma}, \alpha)\}| \geq \delta C_n \right) & = \Pr \left( \sup_{\boldsymbol{\gamma} \in \mathcal{B}_{D}, \alpha \in V_{\sigma}} |v_n(f_{\boldsymbol{\gamma}, \alpha})| \geq \delta n^{1/2} C_n  \right)
\\ & \leq c_0 \exp\left[ - c_0 \delta^2 n /K \right],
\end{align*}
for every $\delta>0$. The exponential tends to zero by our limit assumptions which imply that $K/n \to 0$ as $n \to \infty$. This establishes \eqref{eq:20} and thus completes the proof.
\end{proof}

\end{document}